\documentclass[11pt]{article}

\usepackage[utf8]{inputenc}
\usepackage{amsmath,amsthm,amssymb}
\usepackage{algorithm}
\usepackage[noend]{algcompatible}
\usepackage{latex_abbreviations}
\usepackage[authoryear]{natbib}
\usepackage{hyperref}
\usepackage{cleveref}
\usepackage{mathtools}
\usepackage{thmtools,thm-restate}
\usepackage[cmyk]{xcolor}
\usepackage{subfig}
\usepackage{fullpage}
\allowdisplaybreaks
\newcommand{\err}[1]{\overline{#1}}
\DeclareMathOperator*{\argmax}{argmax}
\DeclareMathOperator*{\argmin}{argmin}

\theoremstyle{plain}
\newtheorem{lemma}{Lemma}[]

\begin{document}

\title{The Secretary Problem with Predictions}

\author{%
Kaito Fujii\\
National Institute of Informatics
\and
Yuichi Yoshida\\
National Institute of Informatics
}

\maketitle

\begin{abstract}
The value maximization version of the secretary problem is the problem of hiring a candidate with the largest value from a randomly ordered sequence of candidates. In this work, we consider a setting where predictions of candidate values are provided in advance. We propose an algorithm that achieves a nearly optimal value if the predictions are accurate and results in a constant-factor competitive ratio otherwise. We also show that the worst-case competitive ratio of an algorithm cannot be higher than some constant $< 1/\mathrm{e}$, which is the best possible competitive ratio when we ignore predictions, if the algorithm performs nearly optimally when the predictions are accurate. Additionally, for the multiple-choice secretary problem, we propose an algorithm with a similar theoretical guarantee. We empirically illustrate that if the predictions are accurate, the proposed algorithms perform well; meanwhile, if the predictions are inaccurate, performance is comparable to existing algorithms that do not use predictions.
\end{abstract}

\section{Introduction}\label{sec:intro}

Online algorithms have been developed to make decisions when information concerning the future is limited.
Considering that recent advances in machine learning have made it possible to predict the future to some extent, it is natural to consider boosting the performance of online algorithms by incorporating these predictions.
In this work, we consider incorporating machine learning predictions into the secretary problem, which is one of the most basic problems in online algorithms.

The \textit{classical secretary problem} models a hiring process where the decision maker interviews candidates one by one in a \textit{uniformly random order}.
The goal of the classical secretary problem is to hire the best candidate from $n$ candidates.
In this study, we consider the value maximization version, in which the decision maker is sequentially provided a real value that represents each candidate's ability and must irrevocably decide whether to hire him/her.
The classical secretary problem has been extensively studied since the 1960s and applied to hiring and many other tasks such as mechanism design \citep{Kle05} and active learning \citep{SH17}.

In the secretary problem, the decision maker knows nothing about the values of the candidates in advance.
However, in realistic scenarios, we can often use prior knowledge about the candidates.
For example, before interviews, it is possible to infer the ability of each candidate by reading his/her r\'esum\'e or consulting public opinion.
Due to the recent development of machine learning for human resources, we can now easily predict the abilities of candidates using public information, and it is natural to consider exploiting these predictions to design better hiring strategies.

In this work, we consider the \emph{secretary problem with predictions}, wherein predictions of candidates' values are given in advance.
The main application we consider is the following employment problem:
Given the r\'esum\'es of the candidates or the public opinion, we predict the abilities of the candidates in advance by utilizing machine learning algorithms.
Note that these information sources might be maliciously manipulated, i.e., the candidates might include lies in their r\'esum\'es to be appreciated more than they really are, or the public opinion gathered on the web might contain false information.
After obtaining the predicted values, the decision maker schedules interviews with each candidate.
Since the order of interviews depends on the personal schedule of each candidate, we can naturally assume that it is a uniform random order.
At the interview of each candidate, the decision maker must immediately decide whether to hire the candidate to prevent the candidate from going to another company.

One extreme approach to this problem is to blindly trust the predictions and hire the candidate with the largest predicted value.
If the predictions are accurate, this approach performs well.
However, if the predictions are adversarially perturbed by malicious agents, this approach performs poorly.
In the above hiring example, if the candidates submit false information for their own interest, the prediction accuracy can arbitrarily deteriorate.
The other extreme approach is to completely ignore the predictions and apply a known algorithm that does not use the predictions.
This approach is not disturbed by adversarially manipulated predictions, but it does not take advantage of the predictions even when they are completely accurate.
We aim to obtain the best of both approaches, that is, an algorithm that achieves better performance when the predictions are sufficiently accurate and the performance comparable to that of the known algorithm even when the predictions are inaccurate.

\subsection{Our contributions}

For the classical secretary problem with predictions, we develop a $\max\{ 0.215, \frac{1-\epsilon}{1+\epsilon} \}$-competitive randomized algorithm, where $\epsilon$ is the largest multiplicative error of the predictions.
We note that our algorithm does not have to know $\epsilon$ in advance.
This competitive ratio converges to $1$ if $\epsilon$ approaches $0$, while it is not worse than a constant of $0.215$ even if $\epsilon$ is large.
Our approach is to modify the well-known $(1/\mathrm{e})$-competitive algorithm~\citep{Dyn63} for the classical secretary problem, which does not use predictions, so that it hires the candidate with the largest predicted value if all the predictions are accurate.
At the arrival of each candidate, the modified algorithm checks whether the predicted value is close to the actual value.
Until the algorithm observes a candidate whose actual value deviates much from the predicted value, it waits for the candidate with the largest predicted value.
If the algorithm has observed such a candidate, it switches to a strategy similar to the algorithm developed by~\cite{Dyn63}.
We show that this modified algorithm achieves a constant competitive ratio even when the predicted values are inaccurate.

On the hardness side, we show that for the classical secretary problem with predictions, any deterministic algorithm cannot be $\max \{ 0.25,  1 - C\epsilon \}$-competitive, and any randomized algorithm cannot be $\max \{ 0.348,  1 - C \epsilon \}$-competitive for any constant $C > 0$.
Since there exists a $1/\rme \approx 0.367$-competitive algorithm for the classical secretary problem without predictions, we can say that it is impossible to simultaneously achieve both $1/\rme$-competitiveness of the classical secretary algorithm and $(1-\rmO(\epsilon))$-competitiveness of the algorithm that hires the best candidate in terms of the predicted values.

Moreover, we apply the same approach to the \textit{multiple-choice secretary problem}, in which the decision maker can hire $k$ candidates, where $k$ is a positive integer given in advance.
As with the classical secretary problem described above, we modify a $(1 - \mathrm{O}(\frac{1}{\sqrt{k}}))$-competitive algorithm~\citep{Kle05} for the multiple-choice secretary problem, which does not use predictions.
We show that the modified algorithm is $(1 - \mathrm{O}(\min\{\frac{\ln k}{\sqrt{k}},\epsilon\} ))$-competitive, which becomes $(1-\mathrm{O}(\epsilon))$-competitive if $\epsilon$ is small, while it achieves almost the same competitive ratio as the original algorithm if $\epsilon$ is large.

\subsection{Outline of this paper}

This paper is organized as follows.
\begin{itemize}
	\item In \Cref{sec:classical}, for the classical secretary problem with predictions, we propose an algorithm with a competitive ratio of $\max \{0.215, \frac{1 - \epsilon}{1 + \epsilon} \}$, where $\epsilon$ is the largest multiplicative error of a predicted value.
	\item In \Cref{sec:hardness}, for the classical secretary problem with predictions, we provide hardness results on the worst-case competitive ratio for deterministic and randomized algorithms.
	\item In \Cref{sec:multiple}, for the multiple-choice secretary problem with predictions, wherein the algorithm is allowed to hire $k$ candidates, we propose an algorithm with a competitive ratio of $1 - \mathrm{O}( \min \{\frac{\ln k}{\sqrt{k}}, \epsilon \})$.
	\item In \Cref{sec:epsilon}, we show that the same competitive ratio bounds hold as in \Cref{sec:classical} and \Cref{sec:multiple} for different definitions of $\epsilon$ that only depend on the prediction errors of relevant candidates.
	\item In \Cref{sec:experiments}, we conducted experiments that corroborate our theoretical results and empirically illustrate the advantages of the proposed algorithms in comparison to existing benchmarks.
\end{itemize}

\subsection{Related work}

The classical secretary problem was proposed by \citet{Gar60}.
\citet{Dyn63} demonstrated a $(1 / \mathrm{e})$-competitive algorithm for the classical secretary problem, which is known to be optimal.
\citet{Lin61} and \citet{GM66} also contributed to the early development of the secretary problem.
We should note that in the original version, the decision maker can only compare the candidates that have already arrived, and the goal is to maximize the probability that the best candidate is hired.

For the multiple-choice secretary problem, \citet{Kle05} proposed an algorithm with a competitive ratio of $1 - \mathrm{O}(\frac{1}{\sqrt{k}})$ and proved that it is asymptotically optimal.
\citet{BIKK07} proposed a $(1 / \mathrm{e})$-competitive algorithm for the same problem, and it was improved for small $k$ values by \citet{AL19}.

A series of studies on prophet secretary~\citep{EHLM17, CFHOV21, EHKS18, ACK18, CSZ21} are similar to our current study.
In the \emph{prophet secretary problem}, in addition to the random order assumption, the decision maker knows the distributions of the candidates' values in advance.
It is known that for several combinatorial constraints, we can achieve a competitive ratio better than that in the secretary problem.
Our algorithms require just a scalar predicted value for each candidate, while algorithms for prophet secretary require the true distributions of the actual values.
Additionally, even when the predicted values are adversarially chosen, our algorithms have constant competitive ratios, while the algorithms for prophet secretary do not provide any such guarantee if the given distributions are even slightly perturbed.

Several existing studies considered variants of the secretary problem similar to ours.
In the problems considered by \citet{CS81}, \citet{KNR20}, and \citet{CCFOT21}, the ground set is randomly partitioned into samples and candidates, and the samples are given in advance to estimate the candidates' values.
Since the samples are randomly selected from the ground set shared with the candidates, they are correlated with the actual values of the candidates.
On the other hand, in our setting, the actual values of the candidates can be adversarially manipulated.

\citet{BGSZ20} and \citet{KM20} considered problem settings in which several adversarial candidates are mixed into random order candidates.
In their setting, the decision maker does not observe any auxiliary information like predictions.
Moreover, their studies interpolate between the random order setting and the adversarial order setting, whereas we interpolate between the random order setting and the trivial setting with known values.

A generic approach to analyzing the competitive ratio for the secretary problem is using linear programming.
\citet{BJS14} proposed linear programming formulations that compute the optimal competitive ratios for the classical secretary problem and its extensions.
\citet{CCJ15} and \citet{DLLV21} extended this LP formulation to more general settings of secretary problems and analyzed the optimal competitive ratios.
In \Cref{sec:hardness}, we solve a linear programming problem to provide an upper bound on the optimal competitive ratios, but our approach is essentially different from these existing approaches.
Since the algorithm is required to hire a nearly optimal candidate if the prediction errors are small,
our LP formulation considers multiple problem instances with the same predicted values and imposes constraints that require the algorithm to hire the candidate with the largest predicted value if a large prediction error is not yet observed,
while the existing LP formulations consider a single problem instance and do not have such a constraint.

\emph{Learning-augmented algorithms} are those based on the idea that we can improve the performance of existing algorithms by using machine learning predictions concerning the input.
This approach has been applied to various problems, including bloom filters~\citep{Mit18}, the ski rental problem~\citep{PSK18,GP19,WZ20}, the caching problem~\citep{LV18,Roh20,Wei20}, branch-and-bound~\citep{BDSV18}, online scheduling~\citep{LLMV20}, metrical task systems~\citep{ACE0S20}, speed scaling~\citep{BMRS20}, and the primal-dual method~\citep{BMS20}.
These studies are similar to ours in that they aim to obtain a better solution if the predictions are accurate while retaining reasonable worst-case guarantees.

\citet{AGKK20} considered the problem setting closest to ours.
In their setting, only the prediction of the maximum value is available, while in our setting, the predictions of all the candidates' values are available.
By using more predicted values, we provide better bounds on the competitive ratio than the algorithm developed by \citet{AGKK20}.
We provide a detailed comparison with their guarantees on the competitive ratio concerning the classical secretary problem with predictions in \Cref{sec:agkk}.

\section{Preliminaries}

The value maximization version of the classical secretary problem is formulated as the problem of hiring a single candidate from $n$ candidates that sequentially arrive.
Let $N = [n] = \{1, \ldots, n\}$ be the set of all candidates.
Each candidate $i \in N$ has a real value $v(i) \in \bbR_{\ge 0}$ that represents his/her ability, but this value is not revealed until the candidate is interviewed.
Initially, the decision maker knows only the number of candidates and nothing about their values.
The candidates in $N$ appear in random order, i.e., the order of the candidates is sampled uniformly at random from the $n!$ permutations.
At the arrival of each candidate, according to his/her value, the decision maker must decide whether to hire this candidate.
This decision is irrevocable, i.e., the decision maker can only hire a candidate right after the interview, and if the decision maker decides to hire the candidate, the hiring process immediately terminates.

In this paper, we consider the learning-augmented setting.
Suppose the algorithm knows the predicted value $\hat{v}(i) \in \bbR_{\ge 0}$ of the actual value $v(i)$ for each $i \in N$ in advance.
At the arrival of each candidate, the algorithm is notified of the index $i \in N$ and the actual value $v(i)$ of the candidate.
Let $\epsilon = \max_{i \in N} \left| 1 - \frac{\hat{v}(i)}{v(i)} \right|$ be the largest multiplicative error among all candidates.
From this definition, we always have $(1 - \epsilon) v(i) \le \hat{v}(i) \le (1 + \epsilon) v(i)$ for all $i \in N$.

In the multiple-choice secretary problem, the capacity $k$ is given in addition to $n$.
As with the classical setting, the decision maker irrevocably decides whether to hire each candidate at his/her arrival and finishes hiring when $k$ candidates are hired or all candidates are interviewed.
By abuse of notation, we write $v(S) = \sum_{i \in S} v(i)$ and $\hat{v}(S) = \sum_{i \in S} \hat{v}(i)$ for $S \subseteq N$.
The goal is to maximize the total value of the hired candidates, $v(S)$, where $S$ is the set of hired candidates.

In both settings, our goal is to design an algorithm that has a theoretical guarantee on its competitive ratio.
We say an algorithm is \emph{$\alpha$-competitive} if the expected value of the output of the algorithm is at least $\alpha$ times the optimal value, where the expectation is taken over the random permutation and the randomness of the algorithm.
Note that our goal is different from the original goal of the classical secretary problem, which is to maximize the probability of hiring the best candidate.
The value maximization problem and the probability maximization problem are equivalent in the original setting, but they are different when predictions can be used.

\subsection{Continuous-time model}
A continuous-time model~\citep{Bruss84, Kle05, FNS11} is a problem setting equivalent to the random order model.
In this model, each candidate $i \in N$ is independently assigned an arrival time $t_i \in [0, 1]$ that is generated from the uniform distribution on $[0, 1]$.
At each time when a candidate $i \in N$ appears at time $t_i$, the decision maker must make an irrevocable decision.

These two models are equivalent, i.e., we can reduce one model to the other model.
We can reduce the random order model to the continuous-time model by independently generating $n$ numbers from the uniform distribution on $[0, 1]$, sorting them in advance, and assigning them as the arrival times to the candidates in order.
The other direction of reduction is straightforward.
As in several existing studies such as \citet{Kle05}, since the continuous-time model is often easier to analyze than the random order model, we adopt the continuous-time model in \Cref{sec:classical} and \Cref{sec:multiple}.
Note that our results can be applied to the random order model by using the reduction mentioned above.

\section{Classical secretary problem}\label{sec:classical}
In this section, we propose an algorithm for the classical secretary problem with predictions.
Our algorithm combines two strategies.
The first strategy is to hire $\hat{\imath} \in \argmax_{i \in N} \hat{v}(i)$, which is a candidate with the top predicted value.
The second strategy is based on Dynkin's algorithm~\citep{Dyn63} in the continuous-time model, which ignores all candidates before time $\tau$ and hires the first candidate after time $\tau$ that is the best thus far.
Dynkin's algorithm in the continuous-time model is $(1/\mathrm{e})$-competitive without predictions if we set $\tau = 1/\mathrm{e}$.
Our proposed algorithm, which we call \textit{learned Dynkin}, initially follows the first strategy but switches to the second strategy if any candidate appears whose multiplicative error is larger than $\theta$.
A detailed description is given in \Cref{alg:classical}.

\begin{algorithm}
	\caption{Learned Dynkin}\label{alg:classical}
	\begin{algorithmic}[1]
		\REQUIRE Time $\tau$, threshold $\theta$, predictions $\hat{v} \colon N \to \bbR$.
		\STATE $\hat{\imath} \in \argmax_{i \in N} \hat{v}(i)$.
		\STATE $\text{mode} \gets \textsc{Prediction}$.
		\FOR{each candidate $i \in N$ in random order}
			\IF{$\displaystyle \left| 1 - \frac{\hat{v}(i)}{v(i)} \right| > \theta$}\label{ln:classical-switch}
				\STATE $\text{mode} \gets \textsc{Secretary}$.
			\ENDIF
			\IF{$\text{mode} = \textsc{Prediction}$ and $i = \hat{\imath}$}
				\STATE Hire $i$.
			\ENDIF
			\IF{$\text{mode} = \textsc{Secretary}$ and $t_i > \tau$ and $i$ is the best so far}
				\STATE Hire $i$.
			\ENDIF
		\ENDFOR
	\end{algorithmic}
\end{algorithm}

If $\theta = 0$ and $\tau = 1/\mathrm{e}$, \Cref{alg:classical} almost coincides with Dynkin's algorithm, which is the best algorithm for the classical secretary algorithm without predictions.
On the other hand, if $\theta$ is sufficiently large, \Cref{alg:classical} always hires $\hat{\imath}$, which is an optimal strategy if the predictions are accurate.
Here we show that if $\theta$ is an intermediate value, \Cref{alg:classical} hires an almost optimal candidate when the predictions are accurate and achieves a constant-factor competitive ratio when the predictions are inaccurate.

Here we define the notations used in the proof.
Let $i^* \in \argmax_{i \in N} v(i)$ be an optimal candidate.
Let $M = \left\{ i \in N \,\middle|\, \left| 1 - \frac{\hat{v}(i)}{v(i)} \right| > \theta \right\}$ be the candidates that significantly deviate from the prediction.
When the first candidate in $M$ appears, the mode of \Cref{alg:classical} switches to \textsc{Secretary}.
Let $m = |M|$.

In the proof, we consider the cases where $M = \emptyset$ and $M \neq \emptyset$ separately.
If $M = \emptyset$, we can easily show that \Cref{alg:classical} follows the predictions and achieves $\frac{1-\epsilon}{1+\epsilon}$-competitiveness.
If $M \neq \emptyset$, we consider the cases where $\hat{\imath} = i^*$ and $\hat{\imath} \neq i^*$ separately, and for each case, we further consider whether $\hat{\imath} \in M$ and whether $i^* \in M$.
For each case, we provide a lower bound on the competitive ratio achieved by \Cref{alg:classical}.

Now we claim that, when $\theta$ and $\tau$ are appropriately chosen, the proposed algorithm is $\frac{1 - \epsilon}{1 + \epsilon}$-competitive if the predictions are accurate, while it achieves a competitive ratio of $0.215$ if the predictions are inaccurate.

\begin{restatable}{theorem}{thmclassical}\label{thm:classical}
	If we set $\theta = 0.646$ and $\tau = 0.313$, then \Cref{alg:classical} is $\max \left\{ 0.215, \frac{1 - \epsilon}{1 + \epsilon} \right\}$-competitive for the classical secretary problem with predictions.
\end{restatable}

\begin{proof}
	First, we consider the case where $M = \emptyset$.
	In this case, since $\epsilon \le \theta$ holds, the mode of the algorithm is \textsc{Prediction} throughout the process, and therefore, the algorithm hires $\hat{\imath}$.
	From the definition of $\epsilon$, we have $(1 - \epsilon) v(i) \le \hat{v}(i) \le (1 + \epsilon) v(i)$ for all $i \in N$.
	Therefore, the value obtained by the algorithm is bounded as follows.
	\begin{equation}\label{eq:empty}
		v(\hat{\imath}) \ge \frac{1}{1+\epsilon} \hat{v}(\hat{\imath}) \ge \frac{1}{1+\epsilon} \hat{v}(i^*) \ge \frac{1-\epsilon}{1+\epsilon} v(i^*).
	\end{equation}
	Hence, the algorithm achieves the competitive ratio $\frac{1-\epsilon}{1+\epsilon}$.
	In this case, since $\epsilon \le \theta = 0.646$, the competitive ratio is at least $\frac{1-\epsilon}{1+\epsilon} \ge \frac{1 - 0.646}{1 + 0.646} \ge 0.215$, which implies $\max \left\{ 0.215, \frac{1 - \epsilon}{1 + \epsilon} \right\}$-competitiveness.

	Next, we consider the case where $M \neq \emptyset$.
	In this case, since $\epsilon > \theta = 0.646$ holds, we have $\frac{1 - \epsilon}{1 + \epsilon} < \frac{1 - 0.646}{1 + 0.646} \le 0.2151$.
	In the remainder of the proof, assuming $\epsilon > \theta$, we prove that the algorithm is $0.2151$-competitive, which implies $\max \left\{ 0.215, \frac{1 - \epsilon}{1 + \epsilon} \right\}$-competitiveness.
	Since $M \neq \emptyset$ in this case, the mode of the algorithm switches from \textsc{Prediction} to \textsc{Secretary} when the first candidate in $M$ appears.
	Let $t^*$ be the time when $i^*$ appears and $\hat{t}$ be the time when $\hat{\imath}$ appears.
	Let $t_M$ be the time when the first candidate in $M$ appears.
	Note that the mode of the algorithm switches from \textsc{Prediction} to \textsc{Secretary} at time $t_M$.

	\begin{enumerate}
		\item[(i)] In the case where $\hat{\imath} = i^* \in M$.

			Since $\hat{\imath} \in M$, the algorithm never hires a candidate when the mode is \textsc{Prediction}.
			Hence, it is sufficient to consider the case where the algorithm hires a candidate when the mode is \textsc{Secretary}.
			When we fix $t^* \in [\tau, 1]$, the algorithm hires $i^*$ if the best candidate in $[0, t^*)$ appears before $\tau$ (this is a sufficient but not a necessary condition).
			If we fix the set of candidates before $t^*$, the best candidate before $t^*$ is also fixed.
			Since the time when this candidate appears conforms to the uniform distribution on $[0, t^*)$, the probability that this time is earlier than $\tau$ is $\tau / t^*$.
			By taking the expectation over $t^*$, we can show that the success probability is at least
				$\int_\tau^1 \frac{\tau}{t^*} \mathrm{d}t^* = \tau \ln \frac{1}{\tau}$.

		\item[(ii)] In the case where $\hat{\imath} = i^* \not\in M$.

			The algorithm hires $\hat{\imath}$ in the \textsc{Prediction} mode when $\hat{\imath}$ appears before all the candidates in $M$.
			Since $M \cup \{\hat{\imath}\}$ appear in random order,
			this probability is $\frac{1}{m+1}$, where $m = |M|$.
			Next, we consider the probability that the algorithm hires $i^*$ in the \textsc{Secretary} mode.
			Let $t^*$ be the time when $i^*$ appears and fix $t^*$.
			When we assume $t^* \in [\tau, 1]$, a sufficient condition for success is that $t_M < t^*$ and the best candidate in $[0, t^*)$ appears before $\tau$.
			Note that, since $\hat{\imath} = i^*$, the algorithm does not hire $\hat{\imath}$ in the \textsc{Prediction} mode when $t_M < t^*$.
			Since the time of each candidate in $M$ conforms to independent uniform distribution on $[0,1]$, the probability that $t_M < t^*$ is $1 - (1 - t^*)^m$.
			The probability that the best candidate in $[0, t^*)$ appears before $\tau$ is $\frac{\tau}{t^*}$.
			By taking the expectation over $t^*$, we can show the success probability is at least
			\begin{align*}
				& \frac{1}{m+1} + \int_\tau^1 \left\{ 1 - {(1 - t^*)}^m \right\} \frac{\tau}{t^*} \mathrm{d}t^* \\
				&= \frac{1}{m+1} + \int_\tau^1 \left\{ 1 - \sum_{k=0}^m {m \choose k} (-t^*)^k \right\} \frac{\tau}{t^*} \mathrm{d}t^* \\
				&= \frac{1}{m+1} + \int_\tau^1 \left\{ 1 - {m \choose 0} (-t^*)^0 - \sum_{k=1}^m {m \choose k} (-t^*)^k \right\} \frac{\tau}{t^*} \mathrm{d}t^* \\
				&= \frac{1}{m+1} + \int_\tau^1 \left\{ - \sum_{k=1}^m {m \choose k} (-t^*)^k \right\} \frac{\tau}{t^*} \mathrm{d}t^* \\
				&= \frac{1}{m+1} + \int_\tau^1 \left\{ \sum_{k=1}^m {m \choose k} (-t^*)^{k-1} \right\} \tau \mathrm{d}t^* \\
				&= \frac{1}{m+1} + \tau \sum_{k=1}^m {m \choose k} \int_\tau^1 (-t^*)^{k-1} \mathrm{d}t^* \\
				&= \frac{1}{m+1} + \tau \sum_{k=1}^m {m \choose k} (-1)^{k-1} \frac{1-\tau^k}{k}.
			\end{align*}

		\item[(iii)] In the case where $\hat{\imath} \neq i^*$, $\hat{\imath} \in M$, and $i^* \in M$.

			Since $\hat{\imath} \in M$, the algorithm switches to the \textsc{Secretary} mode when observing $\hat{\imath}$.
			The algorithm hires $i^*$ if $t^* \in [\tau, 1]$ and any candidate before $i^*$ is not hired.
			Hence, a sufficient condition for hiring $i^*$ is that $t^* \in [\tau, 1]$ and the best candidate in $[0,t^*)$ appears before $\tau$.
			Note that the algorithm might successfully hire $i^*$ if the algorithm observes the best candidate in $[0,t^*)$ after $\tau$ in the \textsc{Prediction} mode, but we ignore this probability.
			The probability that the sufficient condition holds can be computed in the same way as (i) as $\int_\tau^1 \frac{\tau}{t^*} \mathrm{d}t^* = \tau \ln \frac{1}{\tau}$.

		\item[(iv)] In the case where $\hat{\imath} \neq i^*$, $\hat{\imath} \in M$, and $i^* \not\in M$.

			Since $\hat{\imath} \in M$, the algorithm never hires a candidate in the \textsc{Prediction} mode.
			Hence, it is sufficient to consider the probability that the algorithm hires $i^*$ in the \textsc{Secretary} mode.
			Since $i^* \not\in M$ and $i^* \neq \hat{\imath}$, it is necessary that a candidate in $M$ appears before $t^*$, i.e., $t_M \in [0, t^*)$.
			If $t^* \in [\tau, 1]$, $t_M \in [0, t^*)$, and the best candidate in $[0, t^*)$ appears before $\tau$, then the algorithm hires $i^*$.
			In the same way as the analysis of the \textsc{Secretary} mode in case (ii), the success probability is at least
			\begin{equation*}
				\int_\tau^1 \left\{1 - (1 - t^*)^m \right\} \frac{\tau}{t^*} \mathrm{d}t^* = \tau \sum_{k=1}^m {m \choose k} (-1)^{k-1} \frac{1-\tau^k}{k}.
			\end{equation*}

		\item[(v)] In the case where $\hat{\imath} \neq i^*$, $\hat{\imath} \not\in M$, and $i^* \in M$.

			We fix $t^* \in [\tau, 1]$.
			First, we consider the case where all candidates in $M$ appear during $[t^*, 1]$, which occurs with probability $(1-t^*)^{m-1}$.
			In this case, the sufficient condition is that $\hat{\imath}$ also appears after $t^*$.
			The success probability in this case is
			\begin{align*}
				\int_\tau^1 {(1 - t^*)}^m \mathrm{d}t^* = \frac{1}{m+1} {(1-\tau)}^{m+1}.
			\end{align*}

			Next, we consider the case where at least one candidate in $M \setminus \{i^*\}$ appears before $t^*$.
			This occurs with probability $1 - (1 - t^*)^{m-1}$, and the probability density function of $t_M$ is $f(t_M) = (m-1)(1-t_M)^{m-2}$.
			When we fix $t^*$, $\hat{\imath}$ appears after $t^*$ with probability $1-t^*$,
			and $\hat{\imath}$ appears before $t^*$ otherwise.
			If $\hat{t} \ge t^*$, a sufficient condition for success is that the best candidate in $[0, t^*)$ appears before $\tau$.
			The success probability in this case is
			\begin{align*}
				\int_\tau^1 (1-t^*) \left\{ 1 - {(1 - t^*)}^{m-1} \right\} \frac{\tau}{t^*} \mathrm{d}t^*.
			\end{align*}

			If $\hat{t} < t^*$, a sufficient condition for success is that $t_M \in [0, \tau)$, $\hat{t} > t_M$, and the best candidate in $[0, t^*)$ appears before $\tau$.
			Let $i_b$ be the best candidate in $[0, t^*)$ and $t_b$ be the time when $i_b$ appears.
			When we fix $t^* \in [\tau, 1]$, $t_M \in [0, \tau)$ and the set of candidates before $t^*$, the success probability conditioned on $t^*$ and $t_M$ can be bounded in the following three cases separately.
			\begin{itemize}
				\item If $i_b =\hat{\imath}$, a necessary and sufficient condition is $t_M < \hat{t} < \tau$.
				This occurs with probability $\frac{\tau - t_M}{t^*}$.
				\item If $i_b \in M$, a sufficient condition is $\hat{t} \ge t_M$ and $t_b \le \tau$.
				In this case, $t_b$ is equal to $t_M$ if $i_b$ is the first candidate in $M$,
				and $t_b$ conforms to the uniform distribution on $[t_M, t^*)$ otherwise.
				In both cases, the probability that $t_b \le \tau$ holds is at least $\frac{\tau - t_M}{t^* - t_M}$.
				Since $\hat{t}$ conforms to the uniform distribution on $[0, t^*)$, the probability that $\hat{t} \ge t_M$ holds is at least $\frac{t^* - t_M}{t^*}$.
				The success probability is at least
				\begin{align*}
					\frac{t^* - t_M}{t^*} \frac{\tau - t_M}{t^* - t_M} = \frac{\tau - t_M}{t^*}.
				\end{align*}
				\item If $i_b \not\in M \cup \{\hat{\imath}\}$, a sufficient condition is $\hat{t} \ge t_M$ and $t_b \le \tau$. Note that $\hat{t}$ and $t_b$ are independent from $t_M$.
				Since $\hat{t} \ge t_M$ holds with probability $\frac{t^* - t_M}{t^*}$ and $t_b \le \tau$ holds with probability $\frac{\tau}{t^*}$, the success probability is at least
				\begin{align*}
					\frac{t^* - t_M}{t^*} \frac{\tau}{t^*} \ge \frac{\tau - t_M}{\tau} \frac{\tau}{t^*} = \frac{\tau - t_M}{t^*}.
				\end{align*}
			\end{itemize}
			Therefore, the success probability in the case where $t_M < t^*$ and $\hat{t} < t^*$ is at least
			\begin{align*}
				\int_\tau^1 t^* \left[ \int_0^\tau (m-1) {(1 - t_M)}^{m-2} \frac{\tau - t_M}{t^*} \mathrm{d}t_M \right] \mathrm{d}t^*.
			\end{align*}

			In total, the success probability is at least
			\begin{align*}
				&\frac{{(1-\tau)}^{m+1}}{m+1} + \int_\tau^1 \left[ (1-t^*) \left\{ 1 - {(1 - t^*)}^{m-1} \right\} \frac{\tau}{t^*} + \int_0^\tau (m-1) {(1 - t_M)}^{m-2} (\tau - t_M) \mathrm{d}t_M \right] \mathrm{d}t^*\\
				&= \frac{{(1-\tau)}^{m+1}}{m+1} + \tau \ln \frac{1}{\tau} - \tau \int_\tau^1 \frac{{(1-t^*)}^m}{t^*} \mathrm{d}t^* - \frac{1-\tau}{m} \left\{ 1 - {(1 - \tau)}^{m} \right\} \\
				&= \frac{{(1-\tau)}^{m+1}}{m+1} - \tau \sum_{k=1}^m {m \choose k} {(-1)}^k \frac{1 - \tau^k}{k} - \frac{1-\tau}{m} \left\{ 1 - {(1 - \tau)}^{m} \right\}.
			\end{align*}

		\item[(vi)] In the case where $\hat{\imath} \neq i^*$, $\hat{\imath} \not\in M$, and $i^* \not\in M$.

			Since both $\hat{\imath}$ and $i^*$ are not in $M$, the value of $\hat{\imath}$ is a good approximation of the value of $i^*$.
			That is,
			\begin{equation}\label{eq:nonempty}
				v(\hat{\imath}) \ge \frac{1}{1+\theta} \hat{v}(\hat{\imath}) \ge \frac{1}{1+\theta} \hat{v}(i^*) \ge \frac{1-\theta}{1+\theta} \hat{v}(i^*).
			\end{equation}
			The algorithm hires $\hat{\imath}$ if $\hat{t} < t_M$.
			This occurs when $\hat{\imath}$ is the first among $M \cup \{\hat{\imath}\}$, which happens with probability $\frac{1}{m+1}$.

			Next, we consider the probability that the algorithm hires $i^*$.
			Since we are now considering the case where $i^* \not\in M$, it is necessary that at least one candidate in $M$ appears before $i^*$.
			Let $i_b$ be the best candidate in $[0, t^*)$ and $t_b$ the time when $i_b$ appears.
			If $\hat{\imath}$ appears after $t^*$, a sufficient condition is $t_b \le \tau$.
			This occurs with probability $\frac{\tau}{t^*}$.
			If $\hat{\imath}$ appears before $t^*$, in the same way as case (v), we can observe the probability that the algorithm hires $i^*$ is at least
			\begin{align*}
				\int_\tau^1 t^* \left[ \int_0^\tau m {(1 - t_M)}^{m-1} \frac{\tau - t_M}{t^*} \mathrm{d}t_M \right] \mathrm{d}t^*,
			\end{align*}
			where we use the fact that the probability density function of $t_M$ is $f(t_M) = m(1-t_M)^{m-1}$.

			In total, the probability that the algorithm hires $i^*$ is at least
			\begin{align*}
				&\int_\tau^1 \left[ (1-t^*)\{1-{(1-t^*)}^m\}\frac{\tau}{t^*} + t^* \int_0^\tau m {(1 - t_M)}^{m-1} \frac{\tau - t_M}{t^*} \mathrm{d}t_M \right] \mathrm{d}t^*\\
				&=  - \tau \sum_{k=1}^{m+1} {m+1 \choose k} {(-1)}^k \frac{1 - \tau^k}{k} - \frac{1-\tau}{m+1} \left\{ 1 - {(1 - \tau)}^{m+1} \right\}.
			\end{align*}

			Finally, by taking the summation of the contributions by $i^*$ and $\hat{\imath}$, the competitive ratio can be bounded from below by
			\begingroup
			\allowdisplaybreaks[0]
			\begin{align*}
				&\frac{1}{m+1} \frac{1-\theta}{1+\theta}  - \tau \sum_{k=1}^{m+1} {m+1 \choose k} {(-1)}^k \frac{1 - \tau^k}{k} - \frac{1-\tau}{m+1} \left\{ 1 - {(1 - \tau)}^{m+1} \right\}.
			\end{align*}
			\endgroup
	\end{enumerate}

	Now we have obtained the lower bounds on the competitive ratios for all six cases.
	To find approximately optimal $\theta$ and $\tau$, we apply the grid search and numerically compute the minimum competitive ratio among these six cases for each choice of $\theta$ and $\tau$.
	Since the competitive ratios for (ii), (iv), (v), and (vi) depend on $m$, which varies with the problem instance, we compute them only for $m \le 50$.
	As a result, we obtain $\theta = 0.646$ and $\tau = 0.313$.

	Here we show that the competitive ratio is actually at least $0.215$ for the obtained parameters $\theta = 0.646$ and $\tau = 0.313$.
	For the cases (i) and (iii), the competitive ratio is at least $\tau \ln \frac{1}{\tau} \ge 0.363$, which is larger than $0.215$.
	For the case (iv), we can check that the competitive ratio is $\tau(1-\tau) \ge 0.215$ for $m=1$ and it is monotonically non-decreasing in $m$.
	The competitive ratio for the case (ii) is at least its second term, which is equal to the competitive ratio for the case (iv).
	For the case (v), the competitive ratio is lower-bounded by 
	\begin{equation*}
		\tau \ln \frac{1}{\tau} - \tau \int_\tau^1 \frac{{(1-t^*)}^m}{t^*} \mathrm{d}t^* - \frac{1-\tau}{m},
	\end{equation*}
	which is monotonically non-decreasing in $m$.
	Since the numerical computation shows that this lower bound exceeds $0.215$ for $m=6$ and the competitive ratio is at least $0.215$ for $m \le 6$, the competitive ratio is at least $0.215$ for all $m \in \bbZ_{\ge 1}$.
	A similar argument applies to the case (vi).

\end{proof}

\subsection{Comparison with \citet{AGKK20}}\label{sec:agkk}

Recently, a paper that deals with problems similar to ours has been published~\citep{AGKK20}.
In this section, we provide a comparison of our study and the results by \citet{AGKK20} for the classical secretary problem with predictions.
In their setting, only the prediction of the maximum value is available, while in our setting, the predictions of all the candidates' values are available.

The result by \citet{AGKK20} on the classical secretary algorithm with predictions can be summarized as follows.
Let $p^*$ be the prediction of the maximum value $\max_{i \in N} v(i)$.
While we bound the competitive ratio in terms of the largest error of all predicted values, $\epsilon = \max_{i \in N} \left| 1 - \frac{\hat{v}(i)}{v(i)} \right|$, the previous study bounded the competitive ratio in terms of the error of $p^*$, which they denoted by $\eta = | p^* - \max_{i \in N} v(i) |$.
Their algorithm has parameters $0 \le \lambda \le p^*$ and $c \ge 1$.
They proved that its competitive ratio is at least
\begin{equation*}
	\begin{cases}
		\max \left\{ \frac{1}{c \mathrm{e}} , \left[ f(c) \left( \max \left\{ 1 - \frac{\lambda + \eta}{\max_{i \in N} v(i)} , 0 \right\} \right) \right]\right\} & \text{if } 0 \le \eta < \lambda\\
		\frac{1}{c \mathrm{e}} & \text{if } \eta \ge \lambda,
	\end{cases}
\end{equation*}
where $f(c)$ is defined in terms of the two branches $W_0$ and $W_{-1}$ of the Lambert W-function such that $f(c) = \exp(W_0(-1/(c\mathrm{e}))) - \exp(W_{-1}(-1/(c\mathrm{e})))$.

We plot this ratio with our bound for different parameters $c$ and $\lambda$.
We set $\max_{i \in N} v(i) = 1$ so that $\lambda$ is scaled as $0 \le \lambda \le 1$.
The horizontal axis represents the multiplicative error of the predictions, $\epsilon$ or $\eta / \max_{i \in N} v(i)$, and the vertical axis represents the competitive ratio.
Note that $\epsilon \ge \eta / \max_{i \in N} v(i)$ always holds if we set $p^* = \hat{v}(\hat{\imath})$, and therefore, it is easier to bound the competitive ratio in terms of $\epsilon$ than in terms of $\eta / \max_{i \in N} v(i)$.
In \Cref{sec:epsilon}, we will show that the same competitive ratio bound holds even if we replace the definition of $\epsilon$ with $\epsilon = \max \left\{ 1 - \frac{\hat{v}(\hat{\imath})}{v(i^*)}, \frac{\hat{v}(\hat{\imath})}{v(\hat{\imath})} - 1 \right\}$, which is always no larger than the original $\epsilon$, but it is still no less than $\eta / \max_{i \in N} v(i) = \max\left\{ 1 - \frac{\hat{v}(\hat{\imath})}{v(i^*)}, \frac{\hat{v}(\hat{\imath})}{v(i^*)} - 1 \right\}$.

If we set $c = 1$, their algorithm coincides with the classical algorithm.
If we set $c = \frac{1}{0.215\mathrm{e}}$, their worst-case competitive ratio coincides with ours, but their ratio for small prediction errors is worse than ours for any $\lambda$.
If we set $c = 3$, their ratio for small prediction errors approaches ours for each $\lambda$, but their worst-case ratio is worse than ours.
Since $f(c)$ converges to $1$ when $c$ goes to $\infty$, their competitive ratio is always smaller than $1 - \frac{2 \eta}{\max_{i \in N} v(i)}$.

\begin{figure}[t!]
	\centering
	\subfloat[][$c=1$]{\includegraphics[width=.32\hsize]{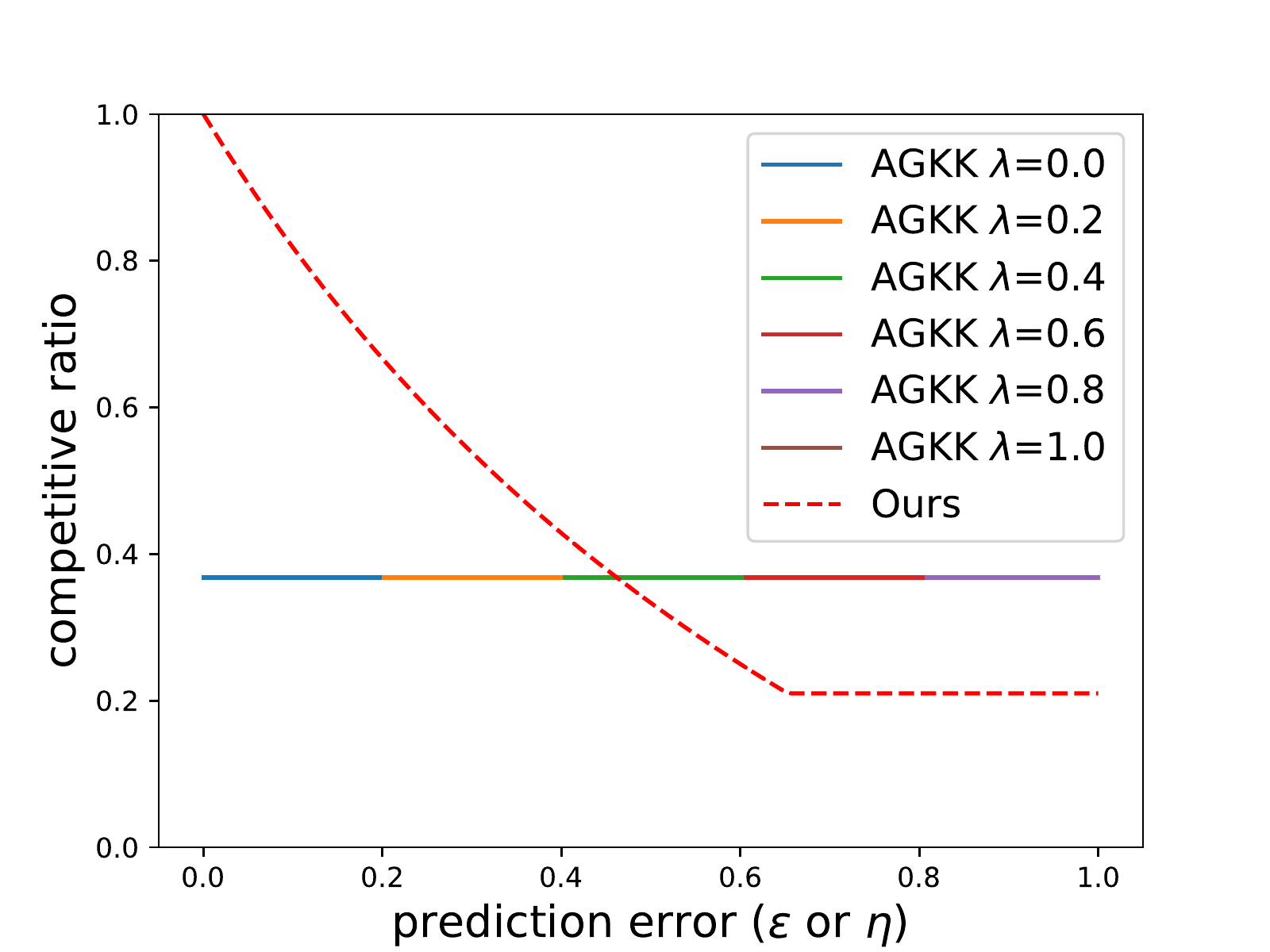}}
	\subfloat[][$c=\frac{1}{0.215\mathrm{e}}$]{\includegraphics[width=.32\hsize]{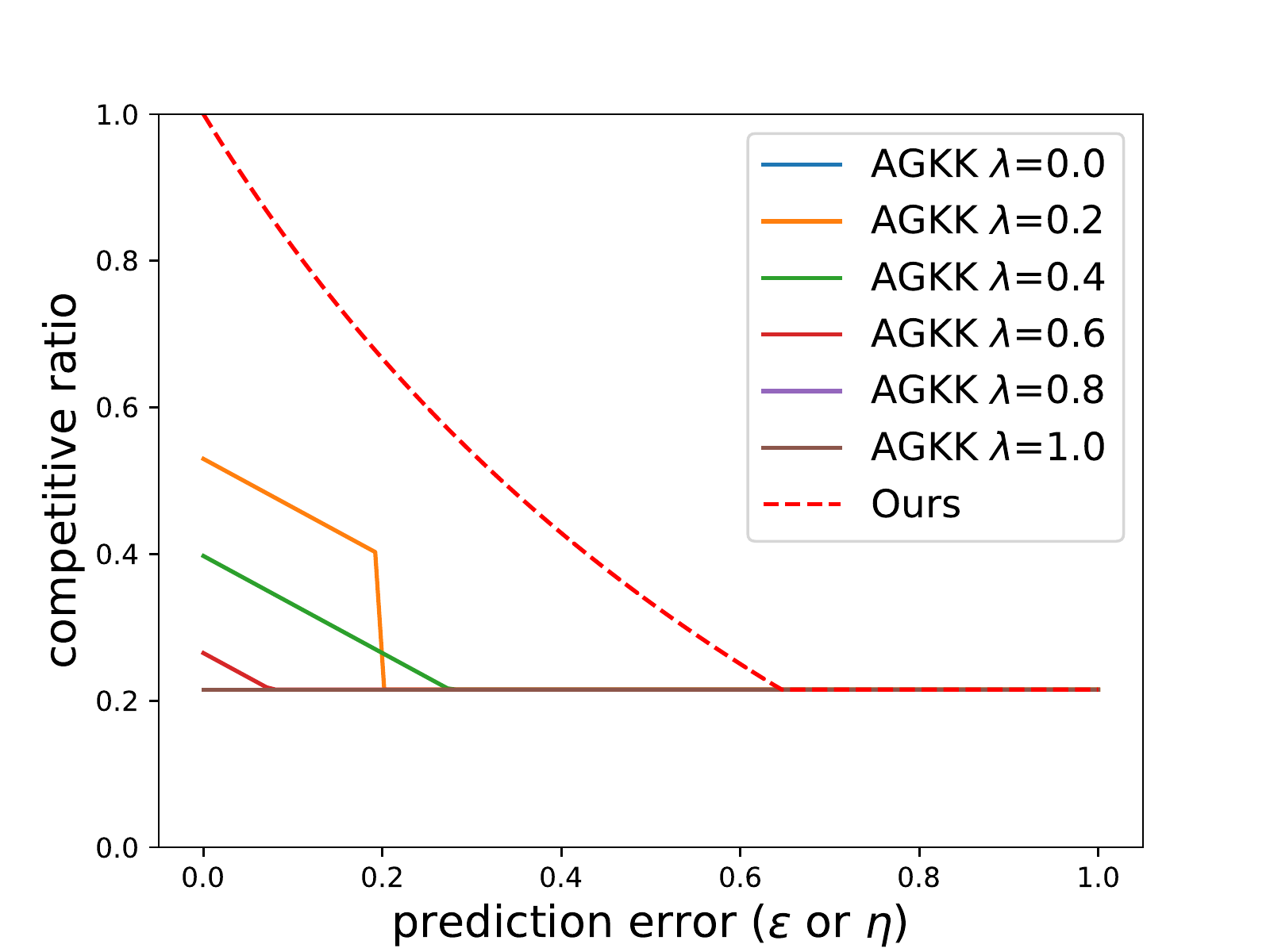}}
	\subfloat[][$c=3$]{\includegraphics[width=.32\hsize]{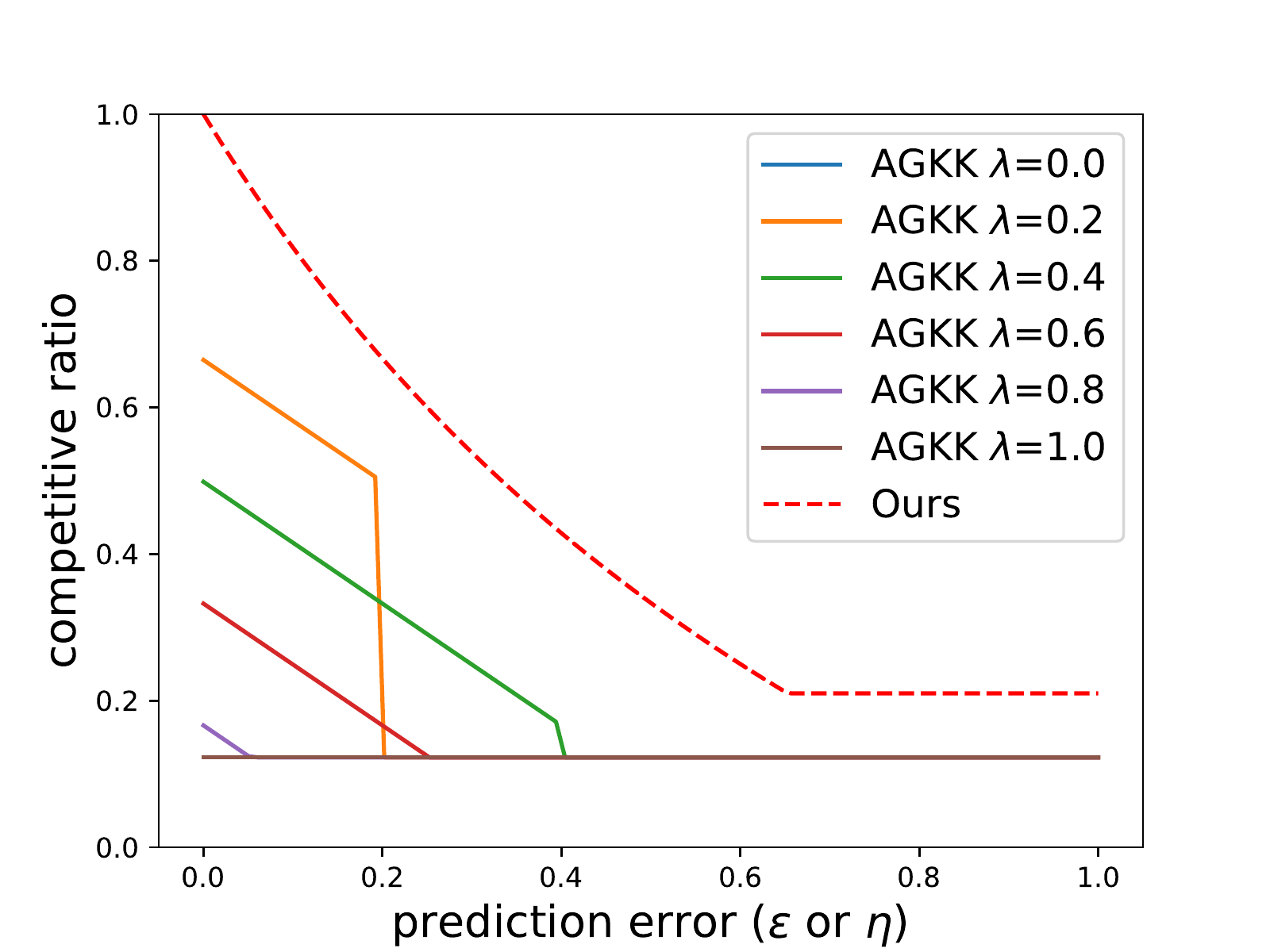}}
	\caption{Comparison of the competitive ratios.\label{fig:comparison}}
\end{figure}

\section{Hardness results for the classical secretary problem}\label{sec:hardness}

In this section, we show hardness results for the classical secretary problem with predictions.

We construct a set of problem instances such that there exists no algorithm that achieves a good competitive ratio for all of them.
We define the set of problem instances that we will use in the proofs as follows.
Let $n$ be the number of candidates, and the set of candidates is represented by $N = \{1,\ldots,n\}$.
The set $\calI_n$ consists of $2^{n-1}$ problem instances, each of which $I_{n,E} \in \calI_n$ is indexed by a subset $E$ of the candidates except $1$, i.e., $E \subseteq \{2,\ldots,n\}$.

Each problem instance is defined as follows.
For each problem instance, the predicted value of candidate $1$ is $L$, where $L$ is a very large value, and it is always correct, i.e., $\hat{v}(1) = v(1) = L$.
For all the problem instances, every other candidate $i \in \{2,\ldots,n\}$ has the same predicted value $\hat{v}(i) = 1$, but their actual values are different depending on the problem instances.
If $i \in E$, then its actual value is $v(i) = L^i$, and if $i \in \{2,\ldots,n\} \setminus E$, then its actual value is $v(i) = 1$.
Intuitively, $E$ represents the set of candidates whose predictions are erroneous.

For any $i \in \{2,\ldots,n\}$, we write candidate $i$ as $i$ when its prediction is accurate and $\err{i}$ when its prediction is inaccurate.
For example, we write $2$ or $3$ when their predictions are accurate, and $\err{2}$ or $\err{3}$ when their predictions are inaccurate.
Here, we use the observation that if all the predicted values and a prefix sequence of actual values of candidates are identical, every algorithm cannot distinguish between problem instances and performs in the same way on that prefix.
Note that the predicted values are identical for all $I_{n,E} \in \calI_n$.
For example, for different problem instances $\{1,\err{2},3\}$ and $\{1,\err{2},\err{3}\}$, if the candidates $1,\err{2}$ arrive in this order, any algorithm cannot distinguish between these two problem instances, and hires $\err{2}$ with the same probability.
Hiring $\err{2}$ is optimal for $\{1,\err{2},3\}$ but not optimal for $\{1,\err{2},\err{3}\}$.

Here, we consider an algorithm that achieves $\max\{1-\rmO(\epsilon), \alpha\}$-competitiveness for some constant $\alpha$, where $\epsilon$ is the largest multiplicative error of predictions.
Since $L$ can be very large, if all the predictions are accurate, i.e., $\epsilon = 0$, an algorithm is required to hire the optimal candidate $1$.
When $1$ appears and all candidates that have appeared so far have accurate predictions, it is possible that all the predictions are accurate, and therefore, the algorithm is required to hire $1$.

If some predictions are not accurate, that is, $E \neq \emptyset$, then $v(i) = L^i$ for any $i \in E \cup \{1\}$ and $v(i) = 1$ for any $i \not\in E$.
In this case, since $\epsilon$ is very large, an algorithms is required to be $\alpha$-competitive.
When we consider a very large $L$, the competitive ratio is arbitrarily close to the probability that the algorithm hires the candidate with the largest index in $E$, and we can ignore the contribution of the other candidates to the competitive ratio.

We provide an upper bound on $\alpha$ such that a deterministic or randomized algorithm is $\max\{1-\rmO(\epsilon), \alpha\}$-competitive for all these problem instances at the same time.

\subsection{A hardness result for deterministic algorithms}

First, we show a hardness result for deterministic algorithms.
Formally, we say an algorithm is deterministic if, given a prefix sequence of candidates observed so far, it deterministically hires or does not hire the next candidate.
Note that we consider deterministic algorithms in the random order model, not the continuous-time model.
Although \Cref{alg:classical} is deterministic in the continuous-time model, we need randomization to adapt it to the random order model, and therefore, it is a randomized algorithm in the random order model.
We show the following.
\begin{restatable}{theorem}{thmhardnessdeterministic}\label{thm:hardness-deterministic}
	For any constant $C > 0$, there exists no deterministic algorithm whose competitive ratio is better than $\max\{1 - C \epsilon, 0.25 \}$ for the classical secretary problem with predictions.
\end{restatable}

\begin{proof}
	We consider the set $\calI_n$ of problem instances defined above in the case of $n = 4$.
	To derive a contradiction, assume there exists an algorithm whose competitive ratio is better than $\max\{1 - C \epsilon, 0.25 \}$.
	If we assume that $L$ is very large compared to $C$, then this algorithm hires the optimal candidate with probability $1$ for $I_{4,\emptyset}$ and with probability larger than $0.25$ for every $I_{4,E} \in \calI_n$ with $E \neq \emptyset$.

	First, we consider $I_{4,\{2\}}$, i.e., the candidates are $\{1, \err{2}, 3, 4\}$.
	There are $24$ orderings of these four candidates.
	Since its competitive ratio is better than $0.25$, the algorithm successfully hires $\err{2}$ for at least $7$ orderings of them.
	As discussed above, since the algorithm achieves $(1-C\epsilon)$-competitiveness for $I_{4,\emptyset}$, the algorithm hires $1$ if $1$ appears before $\err{2}$.
	Hence, if $1$ appears before $\err{2}$, the algorithm fails to hire $\err{2}$.

	There are $12$ orderings out of $24$ orderings in which $1$ appears before $\err{2}$.
	Of the remaining 12 orderings, $\err{2}$ appears first in $6$ orderings, and $\err{2}$ appears second or later in the other $6$ orderings.
	Since the algorithm is deterministic, it hires $\err{2}$ that appears as the first candidate with probability $0$ or $1$.
	Considering that the algorithm hires $\err{2}$ for at least $7$ orderings, the algorithm must hire $\err{2}$ that appears as the first candidate.
	Therefore, even for other problem instances, if $\err{2}$ appears as the first candidate, the algorithm must hire $\err{2}$.

	In the same way, we can show that the algorithm hires $\err{3}$ and $\err{4}$ if they appear as the first candidate.

	Now we consider $I_{4,\{2, 3, 4\}}$, i.e., the candidates are $\{1, \err{2}, \err{3}, \err{4}\}$.
	For this instance, the algorithm hires the first candidate if it is any of $\{1, \err{2}, \err{3}, \err{4}\}$.
	Therefore, the probability that the algorithm hires the optimal candidate $\err{4}$ is $0.25$, which derives a contradiction.

\end{proof}

\subsection{A hardness result for randomized algorithms}

We show that any randomized algorithm cannot achieve better than $\max \{1-\rmO(\epsilon),0.348\}$-competitiveness.
Unlike the deterministic case, we formulate a linear program whose optimal value matches the maximum competitive ratio that can be achieved by a randomized algorithm and solve it by using an LP solver.
Since there exists a $1/\rme \approx 0.367$-competitive algorithm for the classical secretary problem without predictions, we can say that it is impossible to simultaneously achieve both $1/\rme$-competitiveness of the classical secretary algorithm and $(1-\rmO(\epsilon))$-competitiveness of the algorithm that hires the best candidate in terms of the predicted values.

\begin{restatable}{theorem}{thmhardnessrandomized}\label{thm:hardness-randomized}
	For any constant $C > 0$, there exists no randomized algorithm whose competitive ratio is better than $\max\{1 - C \epsilon ,0.348\}$ for the classical secretary problem with predictions.
\end{restatable}

\begin{proof}
	As in the deterministic case, we consider the set $\calI_n$ of problem instances defined above and use the same notation.
	We provide a linear programming formulation whose optimal value is the competitive ratio achieved by an optimal randomized algorithm.

	We define LP variables as follows.
	Let $\Sigma$ be the set of all possible partial permutations.
	Formally, $\Sigma_k = \{ (i_1, \ldots, i_k) \mid i_1, \ldots, i_k \in \{1,\ldots,n\} \cup \{\err{2},\ldots,\err{n}\} ~\text{and}~ | \{i_1,\ldots,i_k\} \cap \{i, \err{i} \} | \le 1 ~ \text{for each $i = 2,\ldots,n$} \}$ for each $k \in [n]$ and $\Sigma = \bigcup_{k=1}^n \Sigma_k$.
	For each $\sigma \in \Sigma$, a variable $x(\sigma)$ represents the probability that an algorithm observes $\sigma$ as a prefix and hires the last candidate of $\sigma$.
	For example, $x(2\err{4}1\err{5})$ represents the probability that the algorithm observes the first four candidates $2, \err{4}, 1, \err{5}$ in this order and hires $\err{5}$.

	First, by induction, we show that the probability that the algorithm observes $\sigma$ is at most
	\begin{equation*}
		\frac{(n-|\sigma|)!}{n!} - \sum_{i=1}^{|\sigma|-1} \frac{(n-|\sigma|)!}{(n-i)!} x(\sigma_i)
	\end{equation*}
	for each $\sigma \in \Sigma$, where $\sigma_i$ represents the sequence of the first $i$ candidates of $\sigma$.
	If $|\sigma| = 1$, the length-$1$ sequence $\sigma$ realizes with probability $1/n$, which proves the base case.
	Suppose the induction hypothesis holds for each $\sigma \in \Sigma$ such that $|\sigma| \le \ell-1$.
	Let $\sigma \in \Sigma$ be any partial permutation of length $\ell$ and $\sigma_{\ell-1}$ be the first $\ell-1$ candidates of $\sigma$.
	By the induction hypothesis, the algorithm observes $\sigma_{\ell-1}$ with probability at most
	\begin{equation*}
		\frac{(n-\ell+1)!}{n!} - \sum_{i=1}^{\ell-2} \frac{(n-\ell+1)!}{(n-i)!} x(\sigma_i).
	\end{equation*}
	The algorithm observes $\sigma_{\ell-1}$ and hires the last candidate in $\sigma_{\ell-1}$ with probability $x(\sigma_{\ell-1})$.
	Furthermore, if the algorithm ignores the last candidate of $\sigma_{\ell-1}$, the next candidate is selected from the $n-\ell+1$ remaining candidates uniformly at random.
	Therefore, the algorithm ignores the last candidate of $\sigma_{\ell-1}$ and observes $\sigma$ with probability
	\begin{equation*}
		\frac{1}{n-\ell+1} \left\{ \frac{(n-\ell+1)!}{n!} - \sum_{i=1}^{\ell-2} \frac{(n-\ell+1)!}{(n-i)!} x(\sigma_i) - x(\sigma_{\ell-1}) \right\}
		= \frac{(n-\ell)!}{n!} - \sum_{i=1}^{\ell-1} \frac{(n-\ell)!}{(n-i)!} x(\sigma_i),
	\end{equation*}
	which proves the induction step.
	Therefore, the probability $x(\sigma)$ that the algorithm observes $\sigma$ and hires the last candidate of $\sigma$ must satisfy
	\begin{equation*}
		x(\sigma) \le \frac{(n-|\sigma|)!}{n!} - \sum_{i=1}^{|\sigma|-1} \frac{(n-|\sigma|)!}{(n-i)!} x(\sigma_i)
	\end{equation*}
	for each $\sigma \in \Sigma$.

	Next, we consider constraints on $x(\sigma)$ for $\sigma$ that contains only $1,\ldots,n$ and ends with $1$.
	If all the predictions are accurate, i.e., $\epsilon = 0$, an algorithm is required to hire the optimal candidate $1$.
	When $1$ appears and all candidates that appeared so far have accurate predictions, it is possible that all the predictions are accurate, and therefore, the algorithm is required to hire $1$.
	In this case, if the problem instance has an inaccurate candidate that appears after $1$, hiring $1$ is not optimal.
	Let $\Sigma_{\text{$1$ must be hired}} = \{ (i_1,\ldots,i_{k-1}, 1) \in \Sigma \mid i_1,\ldots,i_{k-1} \in \{2,\ldots,n\} \}$ be the set of all sequences in which the algorithm must hire $1$.
	Since all candidates in $\{2,\ldots,n\}$ are not optimal and partial permutation $\sigma$ realizes with probability $\frac{(n-|\sigma|)!}{n!}$, we can add a constraint $x(\sigma) = \frac{(n-|\sigma|)!}{n!}$ for all $\sigma \in \Sigma_{\text{$1$ must be hired}}$.

	For each problem instance $I_{n,E} \in \calI_n$, the probability that the algorithm successfully hires an optimal candidate can be represented as follows.
	Let $\Sigma_E \subseteq \Sigma$ be the set of partial permutations that can realize when the erroneous candidates are $E$ such that the last candidate of $\sigma$ is the optimal candidate.
	The probability that the algorithm successfully hires an optimal candidate for problem instance $I_{n,E}$ is $\sum_{\sigma \in \Sigma_E} x(\sigma)$.
	Let $z$ be the competitive ratio of the algorithm.
	Then $z$ is the minimum of these competitive ratios, and therefore, our objective is to maximize $z$ under constraints $\sum_{\sigma \in \Sigma_E} x(\sigma) \ge z$ for each $E \subseteq \{2,\ldots,n\}$.

	Therefore, we can formulate an LP that the competitive ratio of any randomized algorithm must satisfy as follows.
	\begin{alignat*}{4}
	&\text{maximize}   &\ \ &z \\
	&\text{subject to} &    & \sum_{\sigma \in \Sigma_E} x(\sigma) \ge z  & \quad & (E \subseteq \{2,\ldots,n\}), \\
	&                  &    & 0 \le x(\sigma) \le \frac{(n-|\sigma|)!}{n!} - \sum_{i=1}^{|\sigma|-1} \frac{(n-|\sigma|)!}{(n-i)!} x(\sigma_i) & \quad & (\sigma  \in \Sigma),\\
	&                  &    & x(\sigma) = \frac{(n-|\sigma|)!}{n!} & \quad & (\sigma \in \Sigma_{\text{$1$ must be hired}}).
	\end{alignat*}

	On the other hand, for each solution that satisfies these constraints, we can construct an algorithm that observes $\sigma$ and hires the last candidate of $\sigma$ with probability $x(\sigma)$.
	Therefore, the optimal value of this LP coincides with the optimal competitive ratio of a randomized algorithm that always hires $1$ if all the predictions are accurate.

	We solve this LP by using an LP solver.
	Since the number of LP variables is at least exponentially large in $n$, the largest instance that we can solve is the case of $n = 7$.
	The optimal value for the case of $n = 7$ is lower than $0.348$.

\end{proof}

\section{Multiple-choice secretary problem}\label{sec:multiple}

In this section, we propose an algorithm for the multiple-choice secretary problem with predictions, wherein we hire at most $k$ candidates to maximize the sum of their actual values.
The basic principle is the same as that for the classical secretary problem, wherein we combine a strategy that follows predictions and a strategy that works without predictions.
As a strategy that works without predictions, we use Kleinberg's algorithm~\citep{Kle05}.
This algorithm first applies Kleinberg's algorithm with capacity $\ell \coloneqq \lfloor \frac{k}{2} \rfloor$ to the candidates before time $\frac{1}{2}$ recursively.
Next, from the candidates that appear after time $\frac{1}{2}$, it hires candidates whose values exceed $v(i_\ell)$ until the capacity is full, where $i_\ell$ is the $\ell$th largest candidate before time $\frac{1}{2}$.
It has been proved that this algorithm is $(1 - \frac{5}{\sqrt{k}})$-competitive for the multiple-choice secretary problem and there exists no $(1 - o(\frac{1}{\sqrt{k}}))$-competitive algorithm \citep{Kle05}.

Our algorithm, which we refer to as \textit{learned Kleinberg}, first follows the strategy that hires the top $k$ predicted values.
If a candidate appears whose multiplicative error is larger than $\theta$, our algorithm hires it and then applies Kleinberg's algorithm \citep{Kle05} with the remaining capacity to the candidates that appear after this time.
A detailed description is given in \Cref{alg:multiple}.

\begin{algorithm}
	\caption{Learned Kleinberg}\label{alg:multiple}
	\begin{algorithmic}[1]
		\REQUIRE Capacity $k$, threshold $\theta$, predictions $\hat{v} \colon N \to \bbR$.
		\STATE $\hat{S} \in \argmax_{S \subseteq N \colon |S| \le k} \sum_{i \in S} \hat{v}(i)$ and $S \gets \emptyset$.
		\FOR{each candidate $i \in N$ in random order}
			\IF{$\displaystyle \left| 1 - \frac{\hat{v}(i)}{v(i)} \right| > \theta$}\label{ln:multiple-switch}
				\STATE Let $k' \gets k - |S| - 1$.
				\STATE Apply Kleinberg's algorithm with capacity $k'$ to the remaining candidates and obtain $T$.
				\STATE \textbf{return} $S \cup \{i\} \cup T$.
			\ENDIF
			\IF{$i \in \hat{S}$}
				\STATE $S \gets S \cup \{i\}$.
				\IF{$|S| = k$}
					\STATE \textbf{return} $S$.
				\ENDIF
			\ENDIF
		\ENDFOR
		\STATE \textbf{return} $S$.
	\end{algorithmic}
\end{algorithm}

Let $\hat{S}$ be the set of candidates with the top $k$ predicted values and $S^*$ be an optimal solution.
Let $M = \left\{ i \in N \,\middle|\, \left| 1 - \frac{\hat{v}(i)}{v(i)} \right| > \theta \right\}$ be the candidates that significantly deviate from the prediction.
The proof is divided into three parts.
(i) If $M = \emptyset$, then the algorithm hires $\hat{S}$, and in this case, $\hat{S}$ is a good approximation of $S^*$.
(ii) If $|M|$ is large, the algorithm switches to Kleinberg's algorithm in an early phase.
In this case, the output of Kleinberg's algorithm is a good approximation of $S^*$.
(iii) If $|M|$ is small but non-zero, a more involved analysis is required.

\subsection{The case where $M = \emptyset$}

First, we consider the case where $M = \emptyset$.
In this case, since \Cref{alg:multiple} hires the top $k$ predictions $\hat{S}$ and all the predictions are accurate, the output is a good approximation of an optimal solution $S^*$.

\begin{restatable}{lemma}{lemempty}\label{lem:empty}
	If $M = \emptyset$, then \Cref{alg:multiple} is $(1 - 2 \epsilon)$-competitive.
\end{restatable}

\begin{proof}
	From the definition of $\epsilon$, we have $(1 - \epsilon) v(i) \le \hat{v}(i) \le (1 + \epsilon) v(i)$ for all $i \in N$.

	In the case where $M = \emptyset$, the algorithm hires all candidates in $\hat{S}$.
	Since $\hat{S}$ is the set of candidates with the top $k$ predicted values, we obtain
	\begin{align}
		& \sum_{i \in \hat{S}} v(i) \ge \sum_{i \in \hat{S}} \frac{1}{1 + \epsilon} \hat{v}(i) \ge \frac{1}{1 + \epsilon} \sum_{i \in S^*} \hat{v}(i) 
		\ge \frac{1 - \epsilon}{1 + \epsilon} \sum_{i \in S^*} v(i) \ge (1 - 2\epsilon) \sum_{i \in S^*} v(i),\label{eq:multiple-empty}
	\end{align}
	where the first and third inequalities result from the definition of $\epsilon$.
	Therefore, in the case where $M = \emptyset$, the algorithm is $(1 - 2 \epsilon)$-competitive.

\end{proof}

\subsection{The case where $|M| > \sqrt{k} \ln k$}

Next, we consider the case where $|M| > \sqrt{k} \ln k$.
In this case, the algorithm switches from following the prediction to Kleinberg's algorithm when the first element of $M$ appears.
Let $t_M \in [0,1]$ be the time when such a candidate appears.
We denote the candidate that appears at $t_M$ by $i_M$.
Let $Y \subseteq N$ be the candidates that arrive in time $[0, t_M)$ and $Z = (N \setminus Y) \setminus \{i_M\}$.
Note that $M \setminus \{i_M\} \subseteq Z$ holds by definition.

The algorithm hires all of the candidates in $\hat{S}$ that appear before $t_M$.
Therefore, the set of candidates hired by the algorithm before $t_M$ is $S = \hat{S} \cap Y$.
Additionally, the algorithm hires $i_M$ at time $t_M$ except when $|S| = k$.
After $t_M$, the algorithm hires $T$ by applying Kleinberg's algorithm to the remaining candidates $Z$ and the remaining capacity $k' = k - |\hat{S} \cap Y| - 1$.

In the case where $|M| > \sqrt{k} \ln k$, \Cref{alg:multiple} switches to Kleinberg's algorithm in an early phase and the remaining candidates $Z$ contain most of an optimal solution $S^*$.
Therefore, we can show that the output $T$ of Kleinberg's algorithm is a good approximation to the optimal solution $S^*$.
Formally, we can obtain the following bound.

\begin{restatable}{lemma}{lemlarge}\label{lem:large}
	If $|M| > \sqrt{k} \ln k$, then \Cref{alg:multiple} is $(1 - 13 \sqrt{\frac{\ln k}{k}})$-competitive.
\end{restatable}

\begin{proof}
	Since the arrival time of each element in $M$ independently conforms to the uniform distribution on $[0, 1]$, we have
	\begin{align*}
		\Pr\left( t_M < \frac{1}{\sqrt{k}} \right) \ge  1 - {\left( 1 - \frac{1}{\sqrt{k}} \right)}^{\sqrt{k} \ln k} \ge 1 - \frac{1}{k}.
	\end{align*}
	By conditioning on $t_M < \frac{1}{\sqrt{k}}$, each candidate in $\hat{S} \setminus M$ independently appears between $[0, t_M)$ with probability $t_M$ and each candidate in $\hat{S} \cap M$ appears after $t_M$.
	By applying H\oe{}ffding's inequality, we obtain
	\begin{align*}
		\Pr\left(\left|\hat{S} \cap Y\right| \le 2 \sqrt{k \ln k}\right) \ge 1 - \exp(-2 \ln k) \ge 1 - \frac{1}{k}.
	\end{align*}
	That is, the number of candidates hired before $t_M$ is at most $2 \sqrt{k \ln k}$ with probability more than $1 - \frac{1}{k}$.
	In this case, $k' \ge k - 2\sqrt{k \ln k} - 1 \ge k - 3 \sqrt{k \ln k}$, where we use $1$ slot for $i_M$.

	Moreover, by conditioning on $t_M < \frac{1}{\sqrt{k}}$ and $|\hat{S} \cap Y| \le 2 \sqrt{k \ln k}$, each candidate in $S^*$ behaves as follows.
	\begin{itemize}
		\item Each candidate in $(S^* \setminus M) \setminus \hat{S}$ is independently contained in $Z$ with probability at least $1 - \frac{1}{\sqrt{k}}$.
		\item Each candidate in $(S^* \setminus M) \cap \hat{S}$ appears after $t_M$ with probability at least $1 - \frac{1}{\sqrt{k}}$.
			Since we are conditioning on $|\hat{S} \cap Y| \le 2 \sqrt{k \ln k}$, they are not independent, but they appear after $t_M$ with higher probability.
		\item Each candidate in $S^* \cap M$ appears after $t_M$ except $i_M$.
	\end{itemize}
	From these observations, each candidate in $S^* \setminus M$ appears after $t_M$ with probability at least $1 - \frac{1}{\sqrt{k}}$, and $i_M$ is selected from $M$ uniformly at random.
	Therefore, we obtain
	\begin{align}
		\bbE[ v(S^* \cap Z) ]
		\ge \left( 1 - \frac{1}{\sqrt{k}} \right) v(S^* \setminus M) + \left( 1 - \frac{1}{|M|} \right) v(S^* \cap M)
		\ge \left( 1 - \frac{1}{\sqrt{k}} \right) v(S^*). \label{eq:sstar_z}
	\end{align}
	Note that we need to consider only $k$ such that $1 - 13 \sqrt{\frac{\ln k}{k}} \ge 0$.
	Since $k' \ge k - 3 \sqrt{k \ln k}$, the competitive ratio of \citet{Kle05} can be bounded as  $1- \frac{5}{\sqrt{k'}} \ge 1 - \frac{8}{\sqrt{k}}$ for such $k$.
	By conditioning on $t_M < \frac{1}{\sqrt{k}}$ and $|\hat{S} \cap Y| \le 2 \sqrt{k \ln k}$, the output of the algorithm can be evaluated as
	\begin{align*}
		&\bbE[ v(S \cup T \cup \{i_M\}) ]
		\ge \bbE[ v(T) ]
		\ge \left( 1 - \frac{8}{\sqrt{k}} \right) \bbE \left[ \max_{T^* \subseteq S^* \cap Z \colon |T^*| \le k'} v(T^*) \right]\\
		&\ge \left( 1 - \frac{8}{\sqrt{k}} \right) \left( 1 - 3 \sqrt{\frac{\ln k}{k}} \right)\bbE \left[ v(S^* \cap Z) \right]
		\ge \left( 1 - 11 \sqrt{\frac{\ln k}{k}} \right) v(S^*). \tag{From~\eqref{eq:sstar_z}}
	\end{align*}
	By recalling that the events $t_M < \frac{1}{\sqrt{k}}$ and $|\hat{S} \cap Y| \le 2\sqrt{k \ln k}$ occur with probability at least $1 - \frac{2}{k}$, we conclude that the algorithm is $(1 - 13 \sqrt{\frac{\ln k}{k}})$-competitive in the case where $|M| \ge \sqrt{k} \ln k$.

\end{proof}

\subsection{The case where $|M| \le \sqrt{k} \ln k$}

Next, we consider the case where $|M| \le \sqrt{k} \ln k$.
Here we use the same notation as in the case where $|M| > \sqrt{k} \ln k$.

In contrast to the case where $M$ is large, we require a more detailed analysis in the case where $|M| \le \sqrt{k} \ln k$.
Since \Cref{alg:multiple} does not switch to Kleinberg's algorithm in an early phase with high probability, we cannot ignore the candidates hired before time $t_M$.
By fixing $t_M$, we compute the values obtained before $t_M$, on $t_M$, and after $t_M$ separately.
Then, we take the summation of them and take the expectation over $t_M$.

The proof is organized as follows.
First, we prove $k'$, which is the remaining slots for Kleinberg's algorithm, is at least $\sqrt{k} \ln k$ with high probability (\Cref{lem:kprime}).
Next, we show that if $k' \ge \sqrt{k} \ln k$, then the competitive ratio of Kleinberg's algorithm is at least $1-1/\sqrt{k'} = 1-O(\ln k/\sqrt{k})$ in expectation (\Cref{lem:reciprocal}).
By using these two lemmas, we prove the competitive ratio is at least $1-2\theta-O(\ln n/\sqrt{k})$ (\Cref{lem:small}).
We provide the full proof below.

First, we prove the following lemma that claims $k'$ is large enough with high probability.
\begin{lemma}\label{lem:kprime}
	Assume $|M| \le \sqrt{k} \ln k$.
	If we define $k' = k - |S| - 1$, then
	\begin{equation*}
		\Pr(k' \ge \sqrt{k} \ln k) \ge 1 - \frac{3 \ln k}{\sqrt{k}}.
	\end{equation*}
\end{lemma}

\begin{proof}
	We consider the case where $t_M < 1 - \frac{2 \ln k}{\sqrt{k}}$ and the case where $t_M \ge 1 - \frac{2 \ln k}{\sqrt{k}}$ separately as
	\begin{align*}
		\Pr(k' < \sqrt{k} \ln k)
		\le & ~ \Pr\left(t_M < 1 - \frac{2 \ln k}{\sqrt{k}}\right) \Pr\left(k' < \sqrt{k} \ln k \,\middle|\, t_M < 1 - \frac{2 \ln k}{\sqrt{k}}\right)\\
		+ & ~ \Pr\left(t_M \ge 1 - \frac{2 \ln k}{\sqrt{k}}\right) \Pr\left(k' < \sqrt{k} \ln k \,\middle|\, t_M \ge 1 - \frac{2 \ln k}{\sqrt{k}}\right).
	\end{align*}
	Note that $|(\hat{S} \setminus M) \cap Z|$ conforms to the binomial distribution with parameters $|\hat{S} \setminus M|$ and $1 - t_M$, and $k' = |(\hat{S} \setminus M) \cap Z| + |\hat{S} \cap M| - 1$.
	Conditioning on $t_M < 1 - \frac{2\ln k}{\sqrt{k}}$, from H\oe{}ffding's inequality, we can bound the probability that $k' < \sqrt{k} \ln k$ holds as
	\begin{align*}
		\Pr\left(k' < \sqrt{k} \ln k \,\middle|\, t_M < 1 - \frac{2 \ln k}{\sqrt{k}}\right) \le \exp\left(- \frac{2{(\sqrt{k} \ln k)}^2}{k}\right) \le \exp( - 2 \ln k) \le \frac{1}{k^2}.
	\end{align*}
	Since the cumulative distribution function of $t_M$ is $\Pr(t_M \le x) = 1 - {(1 - x)}^{|M|}$, we have $\Pr(t_M \ge 1 - 2 \ln k / \sqrt{k}) \le 2 \ln k / \sqrt{k}$.
	Therefore,
	\begin{align*}
		\Pr(k' < \sqrt{k} \ln k) \le \frac{1}{k^2} + \frac{2 \ln k}{\sqrt{k}} \le \frac{3 \ln k}{\sqrt{k}},
	\end{align*}
	which leads to the statement.

\end{proof}

Note that the competitive ratio $1-\frac{5}{\sqrt{k'}}$ of Kleinberg's algorithm worsens as $k'$ becomes small.
To bound the competitive ratio of Kleinberg's algorithm, we use the following formula concerning the expected value of the reciprocal of the binomial distribution.
\begin{lemma}[{\citep{CS72}}]\label{lem:cs72}
	If $X$ is a random variable that conforms to the binomial distribution with parameters $n$ and $p$, then
	\begin{align*}
		\bbE \left[ \frac{1}{X+1} \right] = \frac{1-{(1-p)}^{n+1}}{(n+1)p}.
	\end{align*}
\end{lemma}

\begin{lemma}\label{lem:reciprocal}
	Conditioning on $k' \ge \sqrt{k} \ln k$, the expected value of $\frac{1}{\sqrt{k'}}$ is bounded as
	\begin{equation*}
		\bbE \left[ \frac{1}{\sqrt{k'}} \,\middle|\, k' \ge \sqrt{k} \ln k \right]
		\le \sqrt{ \left( 1 + \frac{3\ln k}{\sqrt{k}} \right) \frac{\ln (k+2)}{k}}.
	\end{equation*}
\end{lemma}

\begin{proof}
	From Jensen's inequality, we have
	\begin{align*}
		\bbE \left[ \frac{1}{\sqrt{k'}} \,\middle|\, k' \ge \sqrt{k}\ln k \right] \ge \sqrt{\bbE\left[\frac{1}{k'} \,\middle|\, k' \ge \sqrt{k}\ln k \right]}.
	\end{align*}
	From \Cref{lem:cs72}, by fixing $t_M$ and taking the expectation over $k'$, we obtain
	\begin{align*}
		\bbE\left[ \frac{1}{\max\{k', 1\}} \right] \le \bbE\left[ \frac{2}{k'+1} \right] \le \frac{1 - t_M^{k+1}}{(1-t_M)(k+1)} \le \frac{1}{k} \sum_{i=0}^{k+1} t_M^i
	\end{align*}
	Recall that the probability density function of $t_M$ is $m (1-t_M)^{m-1}$.
	Since $\frac{1}{k} \sum_{i=0}^{k+1} t_M^i$ is monotone non-decreasing in $t_M \in [0, 1]$, the expected value is maximized when $m = 1$.
	By taking the expectation over $t_M$, we have
	\begin{align*}
		&\bbE\left[ \frac{1}{\max\{k', 1\}} \right]
		\le \int_0^1 \frac{1}{k} \sum_{i=0}^{k+1} t_M^i m {(1 - t_M)}^{m-1} \mathrm{d}t_M
		\le \int_0^1 \frac{1}{k} \sum_{i=0}^{k+1} t_M^i \mathrm{d}t_M\\
		&= \frac{1}{k} \sum_{i=0}^{k+1} \int_0^1 t_M^i \mathrm{d}t_M
		= \frac{1}{k} \sum_{i=0}^{k+1} \frac{1}{i+1}
		\le \frac{\ln (k+2)}{k}.
	\end{align*}
	By considering
	\begin{align*}
		\bbE\left[\frac{1}{\max\{k', 1\}} \right]
		\ge \Pr(k' \ge \sqrt{k} \ln k) \cdot \bbE\left[\frac{1}{k'} \,\middle|\, k' \ge \sqrt{k} \ln k \right]
		\ge \left( 1 - \frac{3\ln k}{\sqrt{k}} \right) \bbE\left[\frac{1}{k'} \,\middle|\, k' \ge \sqrt{k} \ln k \right], \tag{from \Cref{lem:kprime}}
	\end{align*}
	we obtain
	\begin{align*}
		\bbE \left[ \frac{1}{\sqrt{k'}} \,\middle|\, k' \ge \sqrt{k} \ln k \right]
		\le \sqrt{ \bbE\left[\frac{1}{k'} \,\middle|\, k' \ge \sqrt{k} \ln k \right] }
		\le \sqrt{ \left( 1 + \frac{3\ln k}{\sqrt{k}} \right) \frac{\ln (k+2)}{k}}.
	\end{align*}
\end{proof}

By utilizing these lemmas, we prove the competitive ratio of \Cref{alg:multiple} for the case where $|M| \le \sqrt{k} \ln k$.

\begin{restatable}{lemma}{lemsmall}\label{lem:small}
	If $|M| \le \sqrt{k} \ln k$, then \Cref{alg:multiple} is $(1 - 2 \theta - 13 \ln k / \sqrt{k})$-competitive.
\end{restatable}

\begin{proof}
	We fix the time $t_M$ when the first element of $M$ appears and this element $i_M \in M$.
	The cumulative distribution function of $t_M$ is $F(t_M) = 1 - {(1 - t_M)}^{|M|}$, and $i_M$ is selected from $M$ uniformly at random.
	We bound the values obtained before $t_M$, at $t_M$, and after $t_M$ separately.

	First, we bound the value obtained before $t_M$.
	Since each candidate in $\hat{S} \setminus M$ independently appears before $t_M$ with probability $t_M$ and the algorithm hires all candidates in $\hat{S}$ before $t_M$, the value obtained before $t_M$ can be bounded as
		$\bbE[ v(S) ] \ge t_M v(\hat{S} \setminus M)$.
	Since $\hat{S}$ is the set with the top $k$ predicted values among $N$, the set $\hat{S} \setminus M$ is the set of the top $|\hat{S} \setminus M|$ predicted values among $N \setminus M$.
	If $|\hat{S} \setminus M| \ge |S^* \setminus M|$, then $\hat{v}(\hat{S} \setminus M) \ge \hat{v}(S^* \setminus M)$.
	Otherwise, the average of the predicted values of $\hat{S} \setminus M$ is no less than the average of the predicted values of $S^* \setminus M$.
	Therefore, we obtain
	\begin{align*}
		\hat{v}(\hat{S} \setminus M) \ge \frac{|\hat{S} \setminus M|}{|S^* \setminus M|} \hat{v}(S^* \setminus M) \ge \frac{k - \sqrt{k} \ln k}{k}\hat{v}(S^* \setminus M).
	\end{align*}
	In both cases, we have
		$\hat{v}(\hat{S} \setminus M) \ge \left(1 - \frac{\ln k}{\sqrt{k}} \right) \hat{v}(S^* \setminus M)$.
	From the definition of $M$, it holds that $\left(1 - \theta \right) v(i) \le \hat{v}(i) \le \left(1 + \theta \right) v(i)$ for all $i \in N \setminus M$, which yields
	\begin{align}\label{eq:multiple-small}
		v(\hat{S} \setminus M) \ge \left( 1 - \theta \right) \hat{v}(\hat{S} \setminus M) \ge \left(1 - \theta - \frac{\ln k}{\sqrt{k}} \right) \hat{v}(S^* \setminus M) \ge \left(1 - 2 \theta - \frac{\ln k}{\sqrt{k}} \right) v(S^* \setminus M).
	\end{align}
	Therefore, we bound the value obtained before $t_M$ as
	\begin{align}\label{eq:bound_before_tm}
		\bbE[ v(S) ] \ge t_M \left(1 - 2 \theta - \frac{\ln k}{\sqrt{k}} \right) v(S^* \setminus M).
	\end{align}

	At time $t_M$, the algorithm hires $i_M$ except when $|S| = k$, whose value is $v(i_M)$.

	Next, we bound the value obtained after $t_M$.
	Let $T^* \subseteq Z$ be the optimal solution for the remaining candidates, i.e., $T^* \in \argmax_{T^* \subseteq Z \colon |T^*| \le k'} v(T^*)$.
	From the competitive analysis by \citet{Kle05}, we have
		$v(T) \ge \max \left\{ 1 - \frac{5}{\sqrt{k'}}, 0 \right\} v(T^*)$,
	where we use the fact that the value is always non-negative.

	From now on, we consider the value $\bbE[v(T)]$ under the condition $k' \ge \sqrt{k} \ln k$.
	We denote the event $k' \ge \sqrt{k} \ln k$ by $A$ for short.
	If $k' \ge |S^* \cap Z|$, then $v(T^*) \ge v(S^* \cap Z)$.
	If $k' < |S^* \cap Z|$, we consider adding $(S^* \cap M) \setminus \{i_M\}$ first and then adding candidates in $S^* \setminus M$.
	Since $|S^* \cap M \setminus \{ i_M \}| \le |M| \le \sqrt{k} \ln k \le k'$, we can bound $v(T^*)$ as
	\begin{align*}
		v(T^*) \ge v((S^* \cap M) \setminus \{i_M\}) + \frac{k' - |(S^* \cap M) \setminus \{i_M\}|}{|(S^* \setminus M) \cap Z|} v((S^* \setminus M) \cap Z).
	\end{align*}
	In both cases, it holds that
	\begin{align*}
		v(T^*)
		\ge v((S^* \cap M) \setminus \{i_M\}) + \min \left\{ 1, \frac{k' - |(S^* \cap M) \setminus \{i_M\}|}{|(S^* \setminus M) \cap Z|} \right\}v((S^* \setminus M) \cap Z).
	\end{align*}
	By conditioning on $|S^* \cap Z| = k''$ for any $k''$ such that $|S^* \cap M| \le k'' \le k$ and taking the expectation over $S^* \cap Z$, we obtain
	\begin{align*}
		\bbE \left[ v(T^*) \right]
		&\ge v((S^* \cap M) \setminus \{i_M\}) + \min \left\{ 1, \frac{k' - |(S^* \cap M) \setminus \{i_M\}|}{|(S^* \setminus M) \cap Z|} \right\} \frac{|(S^* \setminus M) \cap Z|}{|S^* \setminus M|} v(S^* \setminus M)\\
		&\ge v((S^* \cap M) \setminus \{i_M\}) + \min \left\{ k'' - |S^* \cap M|, k' - |(S^* \cap M) \setminus \{i_M\}| \right\} \frac{v(S^* \setminus M)}{|S^* \setminus M|}\\
		&\ge v((S^* \cap M) \setminus \{i_M\}) + \left( \min \left\{ k'', k' \right\} - \sqrt{k} \ln k \right) \frac{v(S^* \setminus M)}{k},
	\end{align*}
	where the final inequality is due to $|S^* \cap M| \le |M| \le \sqrt{k} \ln k$ and $|S^* \setminus M| \le |S^*| = k$.
	By combining it with Kleinberg's bound and taking the expectation over $k'$ and $k''$, we obtain
	\begin{align*}
		&\bbE \left[ v(T) \,\middle|\, A \right]\\
		&\ge \bbE \left[ \left(1 - \frac{5}{\sqrt{k'}}\right) \left\{ v((S^* \cap M) \setminus \{i_M\}) + \left( \min \left\{ k'', k'\right\} - \sqrt{k} \ln k \right) \frac{v(S^* \setminus M)}{k} \right\}  \,\middle|\, A \right]\\
		&\ge \left( 1 - \bbE \left[ \frac{5}{\sqrt{k'}} \,\middle|\, A  \right] \right) v((S^* \cap M) \setminus \{i_M\})
		+ \bbE \left[ \min \left\{ k'', k'\right\} - 5 \sqrt{k'} - \sqrt{k} \ln k  \,\middle|\, A \right] \frac{v(S^* \setminus M)}{k}\\
		&\ge \left( 1 - \bbE \left[ \frac{5}{\sqrt{k'}}  \,\middle|\, A \right] \right) v((S^* \cap M) \setminus \{i_M\})
		+ \left( \frac{\bbE\left[\min \left\{ k'', k'\right\}  \,\middle|\, A \right]}{k} - \frac{5}{\sqrt{k}} - \frac{\ln k}{\sqrt{k}} \right) v(S^* \setminus M).
	\end{align*}

	Here we evaluate the expected value of $\min \{k', k'' \}$ under the condition $k' \ge \sqrt{k} \ln k$.
	We can make a coupling between $k'$ and $k'|A$ such that the value of $k'$ is always larger than $k'|A$.
	Similarly, we can make a coupling between $k''$ and $k''|A$ such that the value of $k''$ is always larger than $k''|A$.
	Define $x_+ = \max\{0,x\}$, that is, $x_+ = x$ if $x \ge 0$ and $x_+ = 0$ if $x < 0$ for any $x \in \bbR$.
	Since $\min\{x,y\} \ge z - {(z-x)}_+ - {(z-y)}_+$ for any $x,y,z \in \bbR$, we have
	\begin{align*}
		\bbE\left[ \min \{k', k'' \} \,\middle|\, A \right]
		&\ge (1-t_M)k-1 - \bbE\left[{((1-t_M)k-1 - k')}_+ \,\middle|\, A \right] - \bbE\left[{((1-t_M)k - k'')}_+ \,\middle|\, A \right]\\
		&\ge (1-t_M)k-1 - \bbE\left[{((1-t_M)k-1 - k')}_+ \right] - \bbE\left[{((1-t_M)k - k'')}_+\right]\\
		&\ge (1-t_M)k-1 - \bbE\left[{(\bbE[k']- k')}_+ \right] - \bbE\left[{(\bbE[k''] - k'')}_+\right]\\
		&\ge (1-t_M)k-1 - \bbE\left[\right|\bbE[k']- k'\left|\right] - \bbE\left[\left|\bbE[k''] - k''\right|\right]\\
		&\ge (1-t_M)k-1 - \sqrt{\bbE\left[{(\bbE[k'] - k')}^2\right]} - \sqrt{\bbE\left[{(\bbE[k''] - k'')}^2\right]} \tag{Jensen's inequality}\\
		&\ge (1-t_M)k-1 - \sqrt{k t_M (1-t_M)} - \sqrt{k t_M (1-t_M)}\\
		&\ge (1-t_M)k-1 - 2 \sqrt{k}.
	\end{align*}
	Therefore, the lower bound on $\bbE[v(T)]$ can be expressed as
	\begin{align}
		&\bbE \left[ v(T)  \,\middle|\, k' \ge \sqrt{k} \ln k \right] \nonumber \\
		&\ge \left( 1 - \bbE \left[ \frac{5}{\sqrt{k'}} \,\middle|\, k' \ge \sqrt{k} \ln k \right] \right) v((S^* \cap M) \setminus \{i_M\}) + \left( 1 - t_M - \frac{1}{k} - \frac{7}{\sqrt{k}} - \frac{\ln k}{\sqrt{k}} \right) v(S^* \setminus M) \nonumber \\
		&\ge \left( 1 - \bbE \left[ \frac{5}{\sqrt{k'}} \,\middle|\, k' \ge \sqrt{k} \ln k \right] \right) v((S^* \cap M) \setminus \{i_M\}) + \left( 1 - t_M - \frac{9 \ln k}{\sqrt{k}} \right) v(S^* \setminus M). \label{eq:bound_after_tm}
	\end{align}

	Now we evaluate the output of the algorithm, which is $v(S \cup \{i_M\} \cup T)$.
	Taking the expectation over all random variables including $k'$ and $t_M$, we obtain
	\begin{align*}
		&\bbE\left[ v(S \cup \{i_M\} \cup T) \right]\\
		&\ge \bbE[v(S) ] + \Pr\left( A \right) \bbE\left[ v(T \cup \{i_M\}) \,\middle|\, A \right] \\
		&\ge \bbE[t_M ] \left( 1 - 2 \theta - \frac{\ln k}{\sqrt{k}} \right) v(S^* \setminus M) + \Pr\left( A \right) \bbE\left[ v(T \cup \{i_M\}) \,\middle|\, A \right] \tag{From~\eqref{eq:bound_before_tm}} \\
		&\ge \bbE[t_M ] \left( 1 - 2 \theta - \frac{\ln k}{\sqrt{k}} \right) v(S^* \setminus M)\\
		&\qquad + \Pr\left( A \right) \left( 1 - \bbE \left[ \frac{5}{\sqrt{k'}} \,\middle|\, A \right] \right) v(S^* \cap M) + \Pr\left( A \right) \left( 1 - \bbE[ t_M | A ] - \frac{9 \ln k}{\sqrt{k}} \right) v(S^* \setminus M) \tag{From~\eqref{eq:bound_after_tm}} \\
		&\ge \bbE[t_M] \left( 1 - 2 \theta - \frac{\ln k}{\sqrt{k}} \right) v(S^* \setminus M)\\
		&\qquad + \left( \Pr\left( A \right) - \bbE \left[ \frac{5}{\sqrt{k'}} \,\middle|\, A \right] \right) v(S^* \cap M) + \left( \Pr\left( A \right) - \bbE[ t_M ] - \frac{9 \ln k}{\sqrt{k}} \right) v(S^* \setminus M) \\
		&\ge \bbE[t_M] \left( 1 - 2 \theta - \frac{\ln k}{\sqrt{k}} \right) v(S^* \setminus M)\\
		&\qquad + \left( 1 - \frac{3 \ln k}{\sqrt{k}} - \bbE \left[ \frac{5}{\sqrt{k'}} \,\middle|\, A \right] \right) v(S^* \cap M) + \left( 1 - \frac{3 \ln k}{\sqrt{k}} - \bbE[ t_M ] - \frac{9 \ln k}{\sqrt{k}} \right) v(S^* \setminus M) \tag{From \Cref{lem:kprime}} \\
		&\ge \left( 1 - 2\theta - \frac{13 \ln k}{\sqrt{k}} \right) v(S^* \setminus M) + \left( 1 - \frac{3 \ln k}{\sqrt{k}} - \bbE \left[ \frac{5}{\sqrt{k'}} \,\middle|\, A \right] \right) v(S^* \cap M).
	\end{align*}

	Since the statement holds for any $k$ such that $1 - 2\theta - 13 \ln k / \sqrt{k} < 0$, we need to consider only large $k$ such that $1 - 2\theta - 13 \ln k / \sqrt{k} \ge 0$.
	From \Cref{lem:reciprocal}, we can obtain an upper bound of $\bbE[ \frac{5}{\sqrt{k'}} | A]$ for large $k$ values as
	\begin{equation*}
		\bbE \left[ \frac{5}{\sqrt{k'}} \,\middle|\, k' \ge \sqrt{k} \ln k \right]
		\le 5 \sqrt{ \left( 1 + \frac{3\ln k}{\sqrt{k}} \right) \frac{\ln (k+2)}{k}}
		\le 7 \sqrt{ \frac{\ln k}{k}}.
	\end{equation*}

	Finally, we obtain
	\begin{align*}
		\bbE\left[ v(S \cup \{i_M\} \cup T) \right]
		&\ge \left( 1 - 2 \theta - \frac{13 \ln k}{\sqrt{k}} \right) v(S^* \setminus M) + \left( 1 - \frac{3 \ln k}{\sqrt{k}} - \frac{7 \sqrt{\ln k}}{\sqrt{k}} \right) v(S^* \cap M)\\
		&\ge \left( 1 - 2 \theta - \frac{13 \ln k}{\sqrt{k}} \right) \left( v(S^* \cap M) + v(S^* \setminus M) \right)\\
		&= \left( 1 - 2 \theta - \frac{13 \ln k}{\sqrt{k}} \right) v(S^*).
	\end{align*}
\end{proof}

\subsection{The competitive ratio bound for \Cref{alg:multiple}}

Finally, combining the lemmas for the cases where $M = \emptyset$, $M > \sqrt{k} \ln k$, and $M \le \sqrt{k} \ln k$, we obtain the following result for \Cref{alg:multiple}.

\begin{restatable}{theorem}{thmmultiple}\label{thm:multiple}
	If we set $\theta = \frac{5 \ln k}{\sqrt{k}}$, then \Cref{alg:multiple} is $\left(1 - \min \left\{\frac{21\ln k}{\sqrt{k}}, 5 \epsilon \right\}\right)$-competitive for the multiple-choice secretary problem with predictions.
\end{restatable}

\begin{proof}
	If $5 \epsilon \le \frac{21 \ln k}{\sqrt{k}}$, then $M = \left\{ i \in N \,\middle|\, \left| 1 - \frac{\hat{v}(i)}{v(i)} \right| > \theta \right\} = \emptyset$.
	Hence, in this case, from \Cref{lem:empty}, the competitive ratio of \Cref{alg:multiple} is $1 - 2 \epsilon \ge 1 - 5 \epsilon$.
	If $5 \epsilon > \frac{21 \ln k}{\sqrt{k}}$, then $M \neq \emptyset$.
	In this case, from \Cref{lem:large} and \Cref{lem:small}, \Cref{alg:multiple} is $\left( 1 - \frac{21 \ln k}{\sqrt{k}} \right)$-competitive.

\end{proof}

\section{Different definition of the prediction error}\label{sec:epsilon}
Here, we show that we can provide the same bounds as \Cref{thm:classical} and \Cref{thm:multiple} for different definitions of the prediction error $\epsilon$ in a way similar to the original proofs.

\subsection{Classical secretary problem}
\Cref{thm:classical} still holds even if we replace the definition of $\epsilon$ with
\begin{equation*}
	\epsilon = \max \left\{ 1 - \frac{\hat{v}(\hat{\imath})}{v(i^*)}, \frac{\hat{v}(\hat{\imath})}{v(\hat{\imath})} - 1 \right\},
\end{equation*}
where $i^* \in \argmax_{i \in N} v(i)$ in an optimal candidate and $\hat{\imath} \in \argmax_{i \in N} \hat{v}(i)$ is a candidate with the largest predicted value.
The first term of the maximum represents how large the optimal actual value is compared to the largest predicted value.
The second term of the maximum represents how small the actual value of $\hat{\imath}$ is compared to its predicted value.
The value of this $\epsilon$ is always no larger than that of the original definition of $\epsilon$.

Along with this, we must change the condition under which the mode switches to \textsc{Secretary} (line~\ref{ln:classical-switch} in \Cref{alg:classical}) as follows: $1 - \frac{\hat{v}(\hat{\imath})}{v(i)} \ge \theta$ or $\left( i = \hat{\imath} ~ \text{and} ~ \frac{\hat{v}(\hat{\imath})}{v(\hat{\imath})} - 1 \ge \theta \right)$.
We also replace the definition of $M$ with
\begin{equation*}
M = \left\{ i \in N \,\middle|\, 1 - \frac{\hat{v}(\hat{\imath})}{v(i)} \ge \theta ~ \text{or} ~ \left(i = \hat{\imath} ~ \text{and} ~ \frac{\hat{v}(\hat{\imath})}{v(\hat{\imath})} - 1 \ge \theta \right)\right\}.
\end{equation*}

The same bound as in \Cref{thm:classical} can be shown by slightly modifying the proof.
Specifically, note that the definition of $\epsilon$ and $M$ is only used in the proof for the case $M = \emptyset$ and the case $\hat{\imath} \neq i^*$, $\hat{\imath} \not\in M$, and $i^* \not\in M$ (case (vi) in the proof of \Cref{thm:classical}).
In the case $M = \emptyset$, we have $\hat{v}(\hat{\imath}) \le (1 + \epsilon) v(\hat{\imath})$ and $\hat{v}(\hat{\imath}) \ge (1 - \epsilon) v(i^*)$ from the definition of $\epsilon$.
Therefore, we obtain
\begin{equation*}
	v(\hat{\imath}) \ge \frac{1}{1+\epsilon} \hat{v}(\hat{\imath}) \ge \frac{1-\epsilon}{1+\epsilon} v(i^*),
\end{equation*}
which can be used instead of~\eqref{eq:empty} in the original proof.
In the case $\hat{\imath} \neq i^*$, $\hat{\imath} \not\in M$, and $i^* \not\in M$, we have $\hat{v}(\hat{\imath}) \le (1 + \theta) v(\hat{\imath})$ and $\hat{v}(\hat{\imath}) \ge (1 - \theta) v(i^*)$ from the definition of $M$.
Therefore, we obtain
\begin{equation*}
	v(\hat{\imath}) \ge \frac{1}{1+\theta} \hat{v}(\hat{\imath}) \ge \frac{1-\theta}{1+\theta} v(i^*),
\end{equation*}
which can be used instead of~\eqref{eq:nonempty} in the original proof.
The rest of the proof holds as it is.

\subsection{Multiple-choice secretary problem}
As with the classical setting, \Cref{thm:multiple} still holds even if we replace the definition of $\epsilon$ with
\begin{equation*}
	\epsilon = \max \left\{ 1 - \min_{i \in N \setminus \hat{S}} \frac{\hat{v}(\hat{\imath}_{\min})}{v(i)}, \max_{i \in \hat{S}} \left| \frac{\hat{v}(i)}{v(i)} - 1 \right| \right\},
\end{equation*}
where $\hat{\imath}_{\min} \in \argmin_{i \in \hat{S}} \hat{v}(i)$ is a candidate with the minimum predicted value among the top $k$ predictions.
The value of this $\epsilon$ is always no larger than that of the original definition of $\epsilon$.

Along with this, we must change the condition under which the mode switches to \textsc{Secretary} (Line~\ref{ln:multiple-switch} in \Cref{alg:multiple}) as follows: $\left( i \in N \setminus \hat{S} ~ \text{and} ~ 1 - \frac{\hat{v}(\hat{\imath}_{\min})}{v(i)} \ge \theta \right)$ or $\left( i \in \hat{S} ~ \text{and} ~ \left| \frac{\hat{v}(i)}{v(i)} - 1 \right| \ge \theta \right)$.
We also replace the definition of $M$ with
\begin{equation*}
M = \left\{ i \in N \,\middle|\, \left( i \in N \setminus \hat{S} ~ \text{and} ~ 1 - \frac{\hat{v}(\hat{\imath}_{\min})}{v(i)} \ge \theta \right) ~ \text{or} ~ \left( i \in \hat{S} ~ \text{and} ~ \left| \frac{\hat{v}(i)}{v(i)} - 1 \right| \ge \theta \right) \right\}.
\end{equation*}

The same bound as in \Cref{thm:multiple} can be shown by slightly modifying the proof.
Specifically, note that the definition of $\epsilon$ and $M$ is only used in the proof for the case $M = \emptyset$ (\Cref{lem:empty}) and the case $|M| \le \sqrt{k} \ln k$ (\Cref{lem:small}).
In the case $M = \emptyset$, we have $\hat{v}(i) \le (1 + \epsilon) v(i)$ for each $i \in \hat{S}$ and $\hat{v}(\hat{\imath}_{\min}) \ge (1 - \epsilon) v(i)$ for each $i \in S^*$ from the definition of $\epsilon$.
Therefore, we obtain
\begin{equation*}
	 \sum_{i \in \hat{S}} v(i) \ge \sum_{i \in \hat{S}} \frac{1}{1 + \epsilon} \hat{v}(i) \ge \frac{1 - \epsilon}{1 + \epsilon} \sum_{i \in S^*} v(i) \ge (1 - 2\epsilon) \sum_{i \in S^*} v(i),
\end{equation*}
which can be used instead of~\eqref{eq:multiple-empty} in the proof of \Cref{lem:empty}.
In the case $|M| \le \sqrt{k} \ln k$, we have $(1 - \theta) v(i) \le \hat{v}(i) \le (1 + \theta) v(i)$ for each $i \in \hat{S} \setminus M$ and $\hat{v}(\hat{\imath}_{\min}) \ge (1 - \theta) v(i)$ for each $i \in S^* \setminus M$ from the definition of $M$.
Hence, we have
\begin{equation*}
	v(\hat{S} \setminus M) \ge \frac{1}{1 + \theta} \hat{v}(\hat{S} \setminus M) \ge \left( 1 - \theta \right) \hat{v}(\hat{S} \setminus M).
\end{equation*}
Moreover, if $|\hat{S} \setminus M| \ge |S^* \setminus M|$, then
\begin{align*}
	&\hat{v}(\hat{S} \setminus M)
	= \hat{v}((\hat{S} \cap S^*) \setminus M) + \hat{v}((\hat{S} \setminus S^*) \setminus M)
	\ge \hat{v}((\hat{S} \cap S^*) \setminus M) + |(\hat{S} \setminus S^*) \setminus M| \hat{v}(\hat{\imath}_{\min})\\
	&\ge (1 - \theta) v((\hat{S} \cap S^*) \setminus M) + (1 - \theta) v((S^* \setminus \hat{S}) \setminus M)
	\ge \left(1 - \theta\right) v(S^* \setminus M).
\end{align*}
	Otherwise, we obtain
\begin{align*}
	\hat{v}(\hat{S} \setminus M)
	&= \hat{v}((\hat{S} \cap S^*) \setminus M) + \hat{v}((\hat{S} \setminus S^*) \setminus M)
	\ge \frac{|\hat{S} \setminus M|}{|S^* \setminus M|} \left( \hat{v}((\hat{S} \cap S^*) \setminus M) + |(S^* \setminus \hat{S}) \setminus M| \hat{v}(\hat{\imath}_{\min}) \right)\\
	&\ge \frac{k - \sqrt{k} \ln k}{k} \left( (1 - \theta) v((\hat{S} \cap S^*) \setminus M) + (1 - \theta) v((S^* \setminus \hat{S}) \setminus M) \right) \tag{since $|M| \le \sqrt{k} \ln k$}\\
	&= \left(1 - \frac{\ln k}{\sqrt{k}} \right) \left(1 - \theta\right) v(S^* \setminus M)
	\ge \left(1 - \theta - \frac{\ln k}{\sqrt{k}} \right) v(S^* \setminus M).
\end{align*}
Therefore, in both cases, we have
\begin{equation*}
	\hat{v}(\hat{S} \setminus M) \ge \left(1 - \theta - \frac{\ln k}{\sqrt{k}} \right) v(S^* \setminus M).
\end{equation*}
Combining these inequalities, we obtain
\begin{equation*}
	v(\hat{S} \setminus M) \ge \left( 1 - \theta \right) \hat{v}(\hat{S} \setminus M) \ge \left(1 - 2 \theta - \frac{\ln k}{\sqrt{k}} \right) v(S^* \setminus M),
\end{equation*}
which can be used instead of~\eqref{eq:multiple-small} in the original proof.
The rest of the proof holds as it is.

\section{Experiments}\label{sec:experiments}

In this section, we present our experimental results.
The goal of the experiments is to illustrate for what kind of problem instances our proposed algorithms perform well, not to show they perform well for all problem instances.
It is because there is a tradeoff between the case of accurate predictions and the case of inaccurate predictions, and then no algorithm can achieve a good competitive ratio for every problem instance.
This tradeoff can be seen from our hardness result (\Cref{thm:hardness-randomized}), which claims any algorithm that performs better for accurate predictions performs worse than the no-prediction algorithm for inaccurate predictions.

More specifically, the experiments are aimed at showing that our proposed algorithms have the following two advantages compared with existing algorithms:
(i) it hires the optimal candidates if the predictions are accurate
and (ii) its worst-case performance is guaranteed even if the predictions are completely wrong.
To verify these advantages, we apply our proposed algorithms to three types of randomly generated problem instances and compare their performance with benchmark algorithms.

\subsection{Experimental settings}
All the algorithms are implemented in Python 3.8.
We conducted the experiments in a machine with Intel Core i7 (2.3 GHz and 4 cores) and 32 GB RAM.

\paragraph{Datasets.}
In all instances, we set $n = 100$.
Let $\epsilon > 0 $ be an error parameter and $k$ be the number of candidates to be selected.
We considered 11 different values of $\epsilon \in \{0,0.1,\dots,0.9,1\}$ and three different values of $k$: $k = 1$ (the classical setting), $k = 10$, and $k = 50$.

To compare the proposed algorithms with benchmarks, we generated three types of problem instances: \textsf{Uniform}, \textsf{Adversarial}, and \textsf{Almost-Constant}.
Each problem instance is equipped with different types of predictions: uniformly perturbed predictions, adversarially perturbed predictions, and constant predictions irrelevant to actual values.
We design these problem instances so that the algorithm that selects top-$k$ predictions performs well for \textsf{Uniform}, poorly for \textsf{Adversarial}, and completely random for \textsf{Almost-Constant}.
\begin{itemize}
\item In \textsf{Uniform}, we sample each $v(i)$ independently from the exponential distribution with parameter $1$.
Since the exponential distribution is an unbounded distribution that generates a large value with a small probability, the optimal value is almost always by far the largest.
Hence, the difference between good candidates and bad candidates becomes distinct, and therefore, the algorithmic performance is clearly expressed in the empirical competitive ratio.
The predicted value is generated by perturbing the actual value with the uniform distribution, i.e., $\hat{v}(i) = \delta_i v(i)$, where $\delta_i$ is sampled uniformly and independently from $[1-\epsilon,1+\epsilon]$.
\item In \textsf{Adversarial}, we consider a situation where the adversary tries to confuse large values and small values while the multiplicative error is at most $\epsilon$.
We sample each actual value $v(i)$ independently from the exponential distribution with parameter $1$ as in \textsf{Uniform}.
Then we adversarially perturbed their predicted values, that is, we set $\hat{v}(i) = (1-\epsilon)v(i)$ for the candidates with top $50$ actual values and $\hat{v}(i) = (1+\epsilon)v(i)$ for the other $50$ candidates.
\item
In \textsf{Almost-Constant}, all the candidates have almost the same predicted values, but $k$ candidates selected uniformly at random have large actual values.
Formally, all the predicted values are set to $1$.
Then, the actual values of randomly selected $k$ candidates are set to $1 / (1-\epsilon)$, while the others' actual values are set to $1$.
Since $1 / (1-\epsilon)$ cannot be defined, we omit the case of $\epsilon = 1$ only for \textsf{Almost-Constant}.
These predicted values and actual values are perturbed with a random value generated from the uniform distribution over $[0, 0.01]$ for randomly breaking ties.
\end{itemize}

We generated 100 datasets of actual and predicted values for each of \textsf{Uniform}, \textsf{Adversarial}, and \textsf{Almost-Constant}.
For each dataset, we estimated the competitive ratio by running each algorithm on the randomly permuted instance 100 times and considered the average value.

\paragraph{Benchmarks.}
For each setting, we compare the proposed algorithm with varying parameter $\theta \in \{0.1,0.3,0.5,0.7,0.9\}$ to the following benchmarks.
\begin{itemize}
\item We plotted the competitive ratios of the original Dynkin or Kleinberg that does not use predictions as ``w/o pred.''
\item We also plotted the value obtained by selecting the top $k$ elements in terms of the predicted values (``top-$k$ pred.'' in the figure).
\item 
For the classical setting, we implemented the algorithm proposed by \citet{AGKK20} for the setting where only a prediction of the optimal value is provided (``AGKK'' in the figure).
We use $\max_{i \in N} \hat{v}(i)$ as the prediction of the optimal value.
This algorithm has two parameters $c \in [1, \infty)$ and $\lambda \in [0, p^*]$, where $p^*$ is the prediction of the optimal value.
Since this experiment is aimed at evaluating the performance for accurate predictions when the worst-case competitive ratio is guaranteed, we set $c = 1.71 \approx \frac{1}{0.215\mathrm{e}}$ so that the worst-case competitive ratio $1/(ce)$ is almost equal to that of learned Dynkin.
Among possible choices of $\lambda \in [0,p^*]$, we select two values $\lambda  \in \{ 0.3p^*, 0.7p^* \}$.
\item
For the classical setting, we implemented an algorithm proposed by \citet{CSZ19} for the prophet secretary problem (``prophet secretary'' in the figure).
The prophet secretary problem assumes that the exact probability distributions of the actual values of all candidates are provided.
Among several thresholding techniques proposed in \citet{CSZ19}, we use one that can be easily computed, whose thresholds are $\alpha(t) = 0.53 - 0.38t$ times the expected maximum for each arrival time $t \in [0, 1]$, which is $0.657$-competitive for the prophet secretary problem.
Since the exact probability distributions are not available in our setting, we apply this algorithm as if the provided probability distribution for each candidate $i \in N$ is the uniform distribution over $[\hat{v}(i) - \theta, \hat{v}(i) + \theta]$, where $\theta$ is a parameter we choose.
We choose the uniform distribution because it is difficult to efficiently compute the thresholds for complicated distributions.
The threshold $\tau_t$ at time $t$ is defined so that $\Pr(\max_{i \in [n]} v(i) \le \tau_t) = \alpha(t)$.
To obtain this value, we need the quantile of the cumulative distribution function of $\max_{i \in [n]} v(i)$, but it is difficult for complicated distributions.
We show results for two choices of $\theta \in \{ 0.3p^*, 0.7p^* \}$, where $p^* = \max_{i \in N} \hat{v}(i)$.
\end{itemize}

\begin{figure*}[t!]
	\centering
	\subfloat[][\textsf{Uniform}, $k=1$\label{fig:uniform_1}]{\includegraphics[height=.25\hsize]{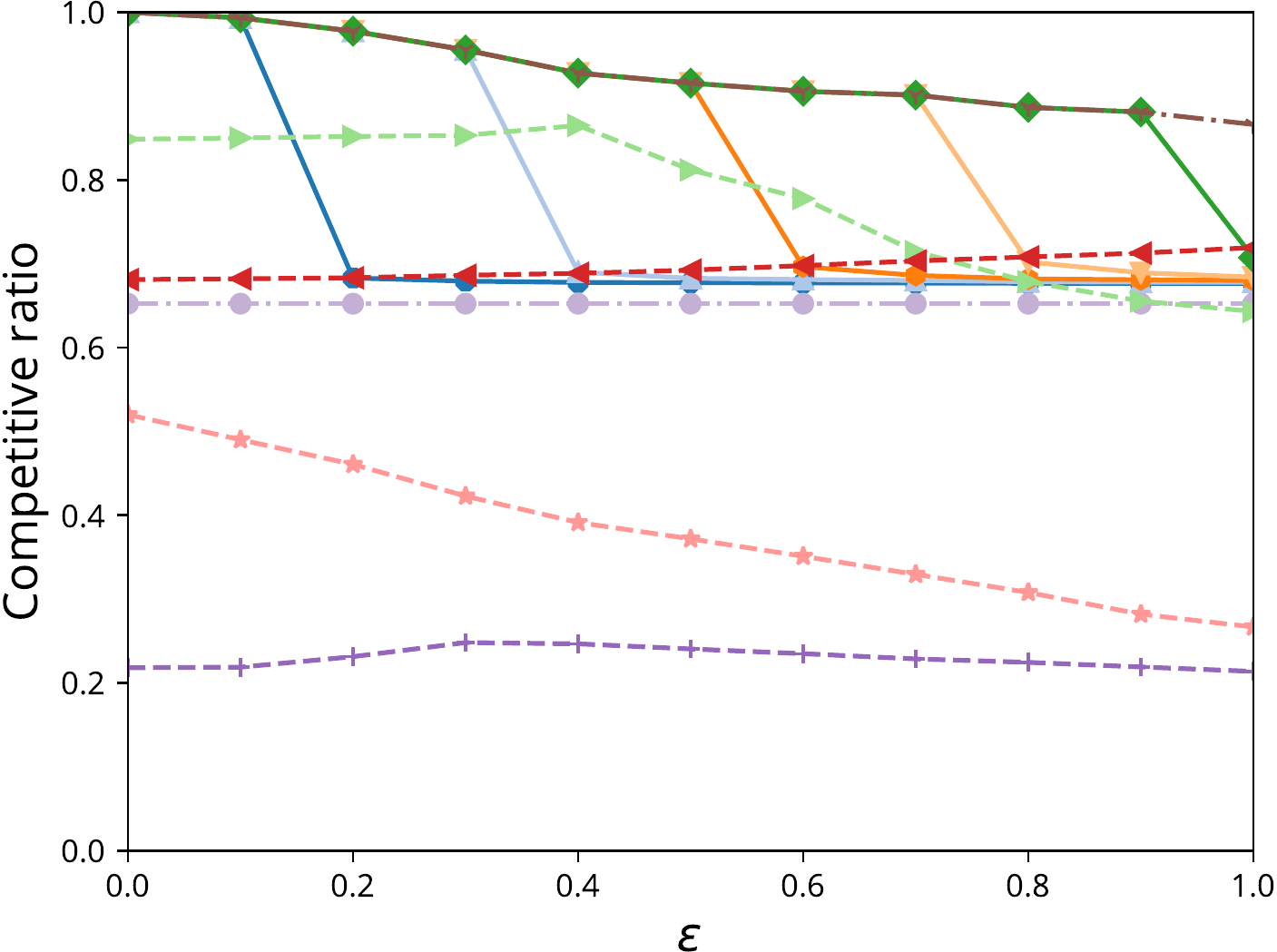}}
	\subfloat[][\textsf{Adversarial}, $k=1$\label{fig:pathological_1}]{\includegraphics[height=.25\hsize]{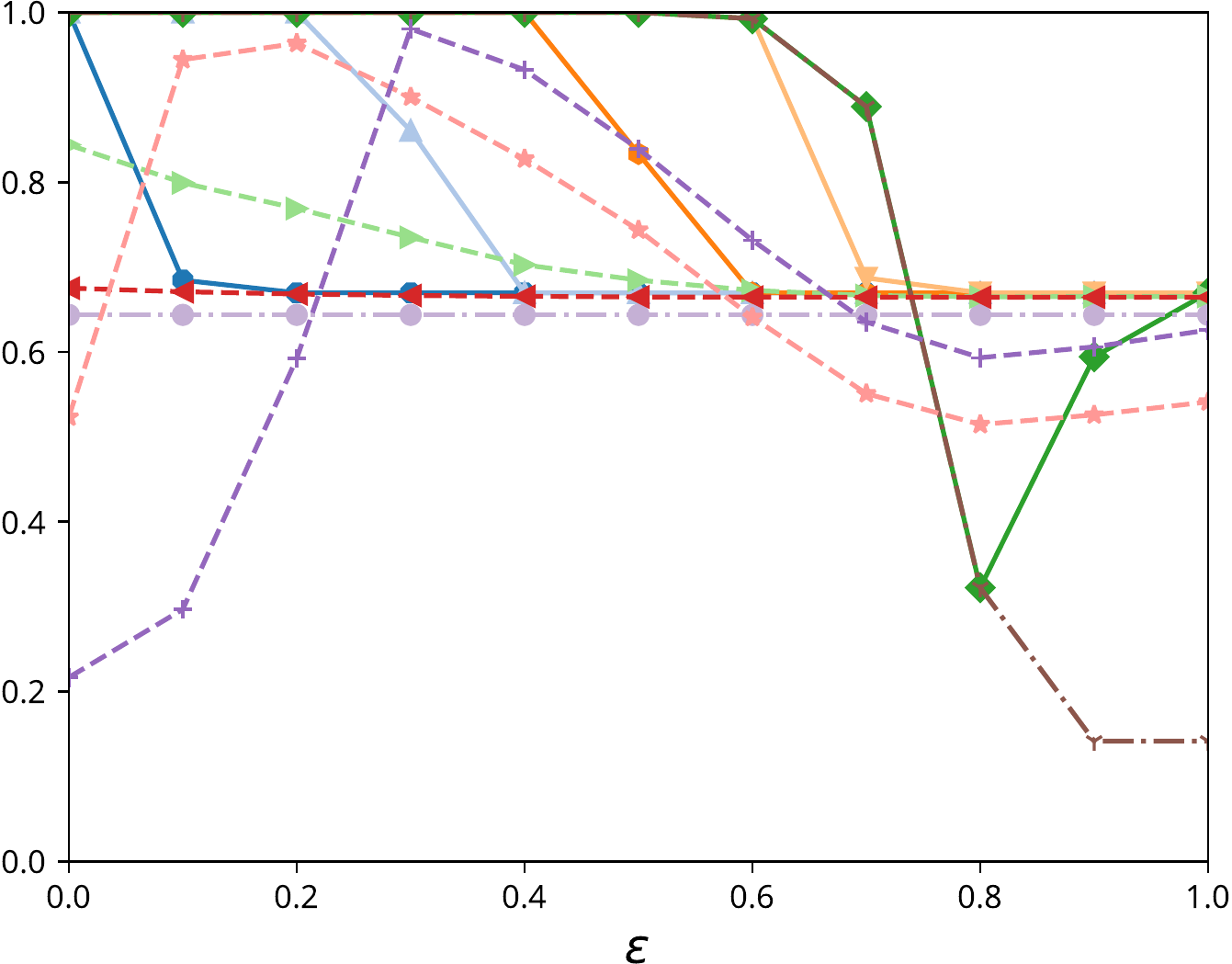}}
	\subfloat[][\textsf{Almost-Constant}, $k=1$\label{fig:permuted_1}]{\includegraphics[height=.25\hsize]{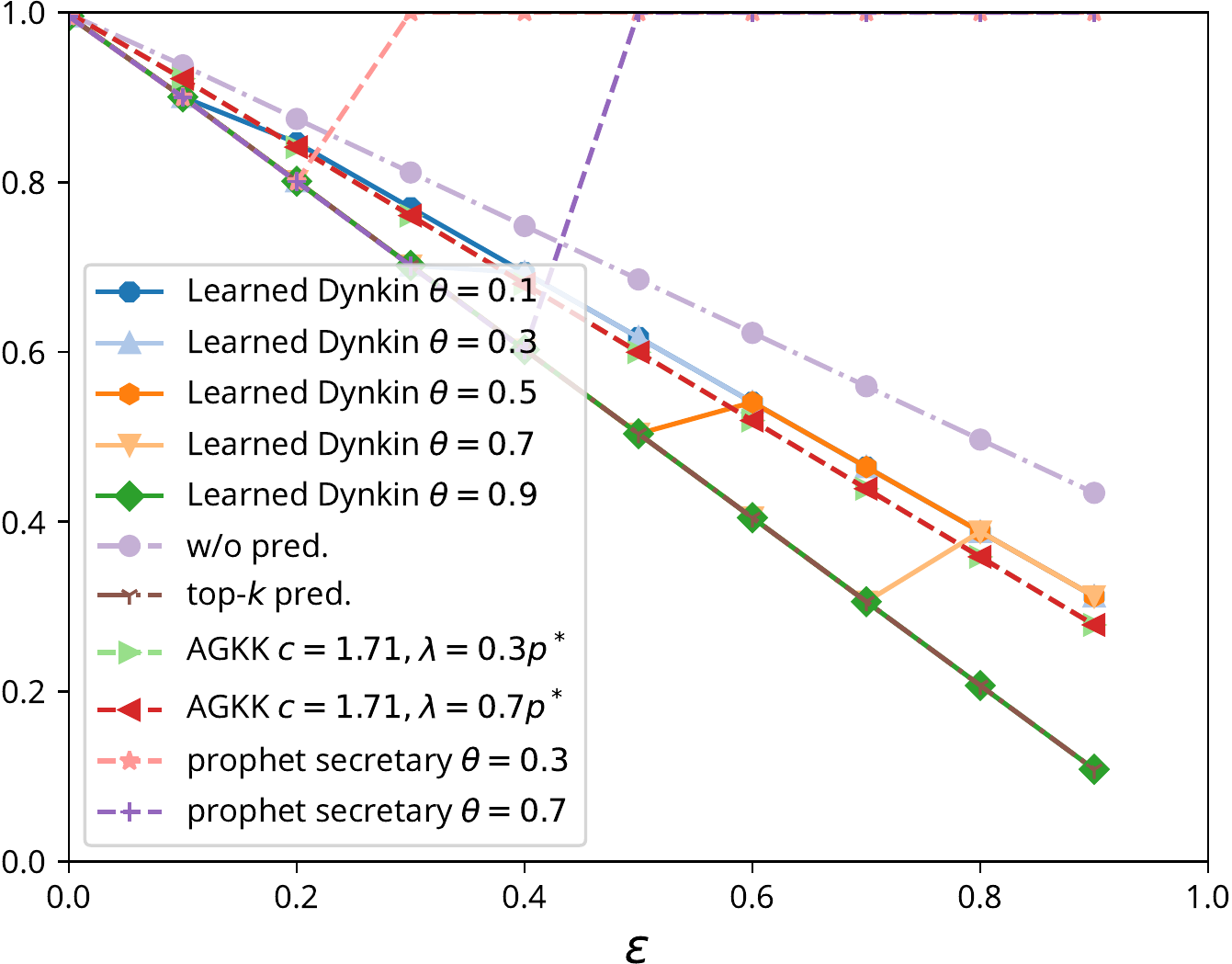}}\\
	\centering
	\subfloat[][\textsf{Uniform}, $k=10$]{\includegraphics[height=.25\hsize]{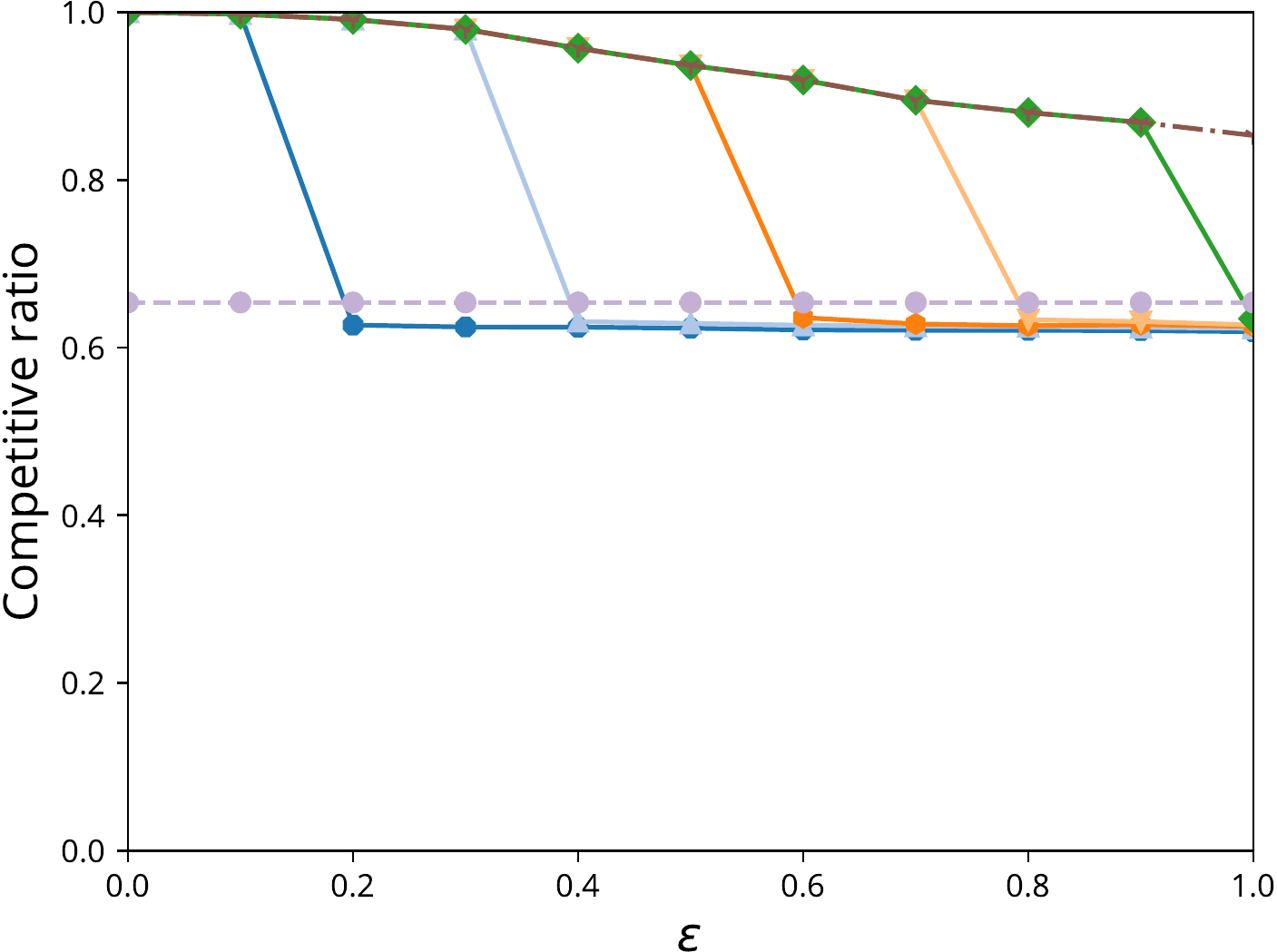}}
	\subfloat[][\textsf{Adversarial}, $k=10$]{\includegraphics[height=.25\hsize]{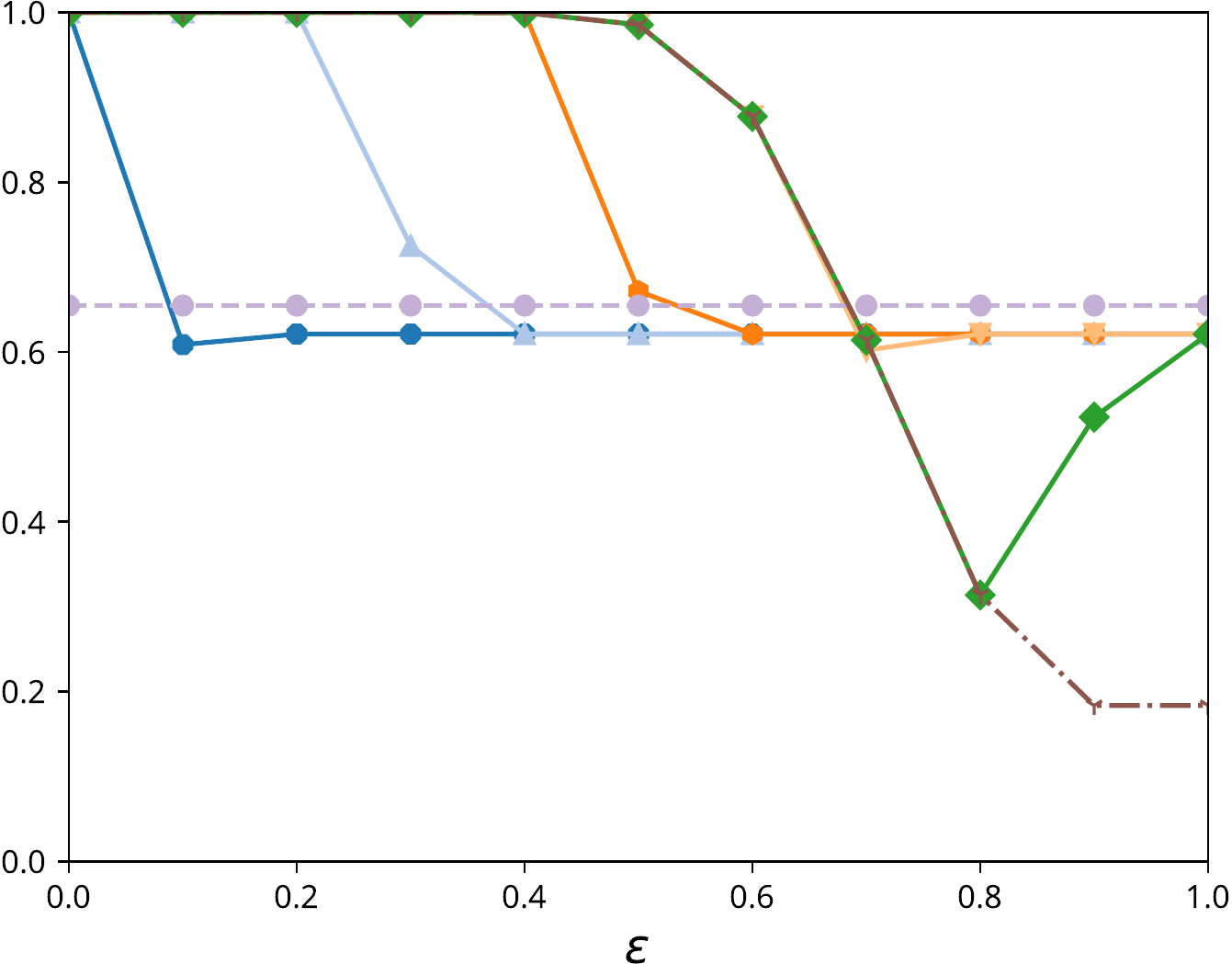}}
	\subfloat[][\textsf{Almost-Constant}, $k=10$]{\includegraphics[height=.25\hsize]{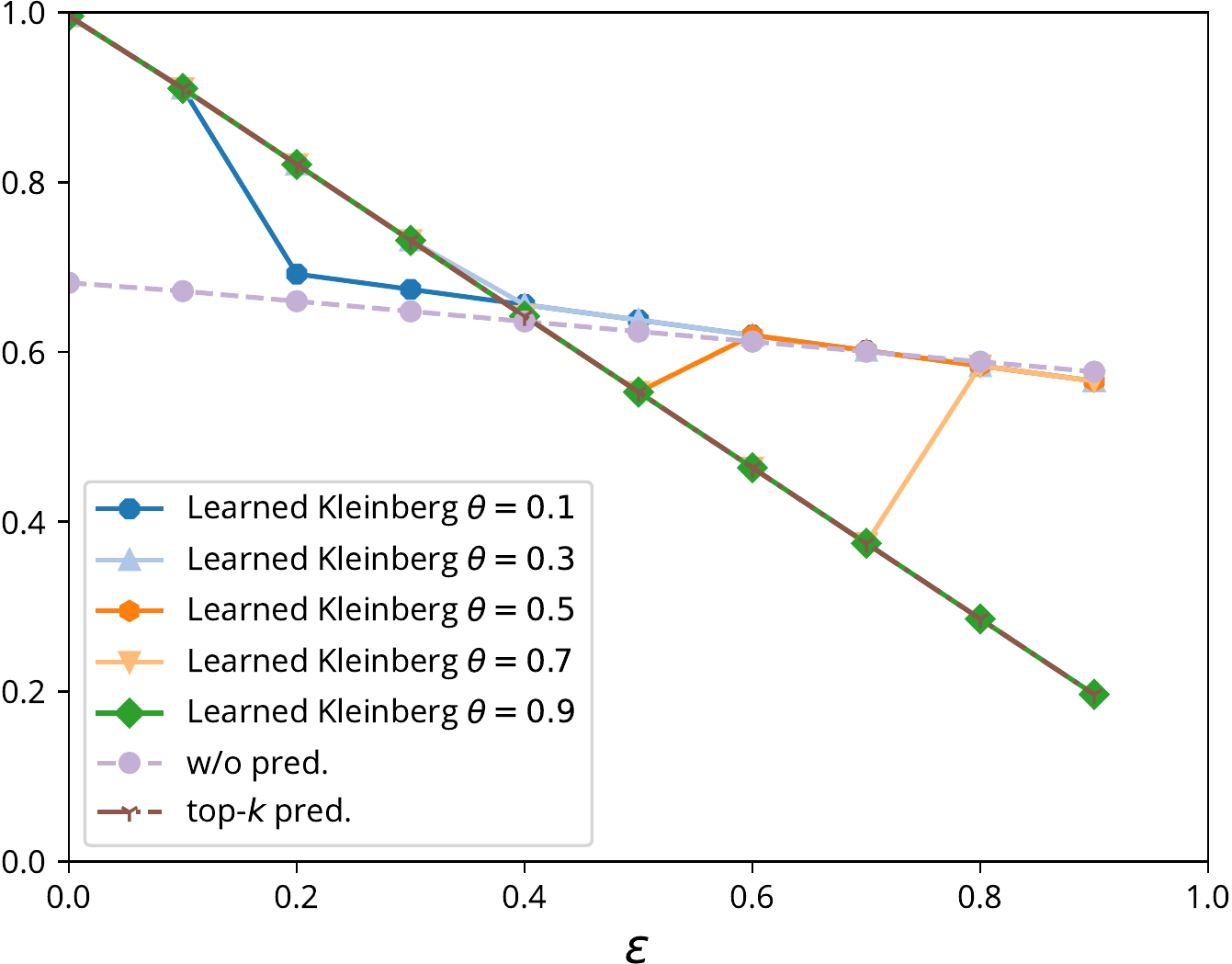}}\\

	\centering
	\subfloat[][\textsf{Uniform}, $k=50$]{\includegraphics[height=.25\hsize]{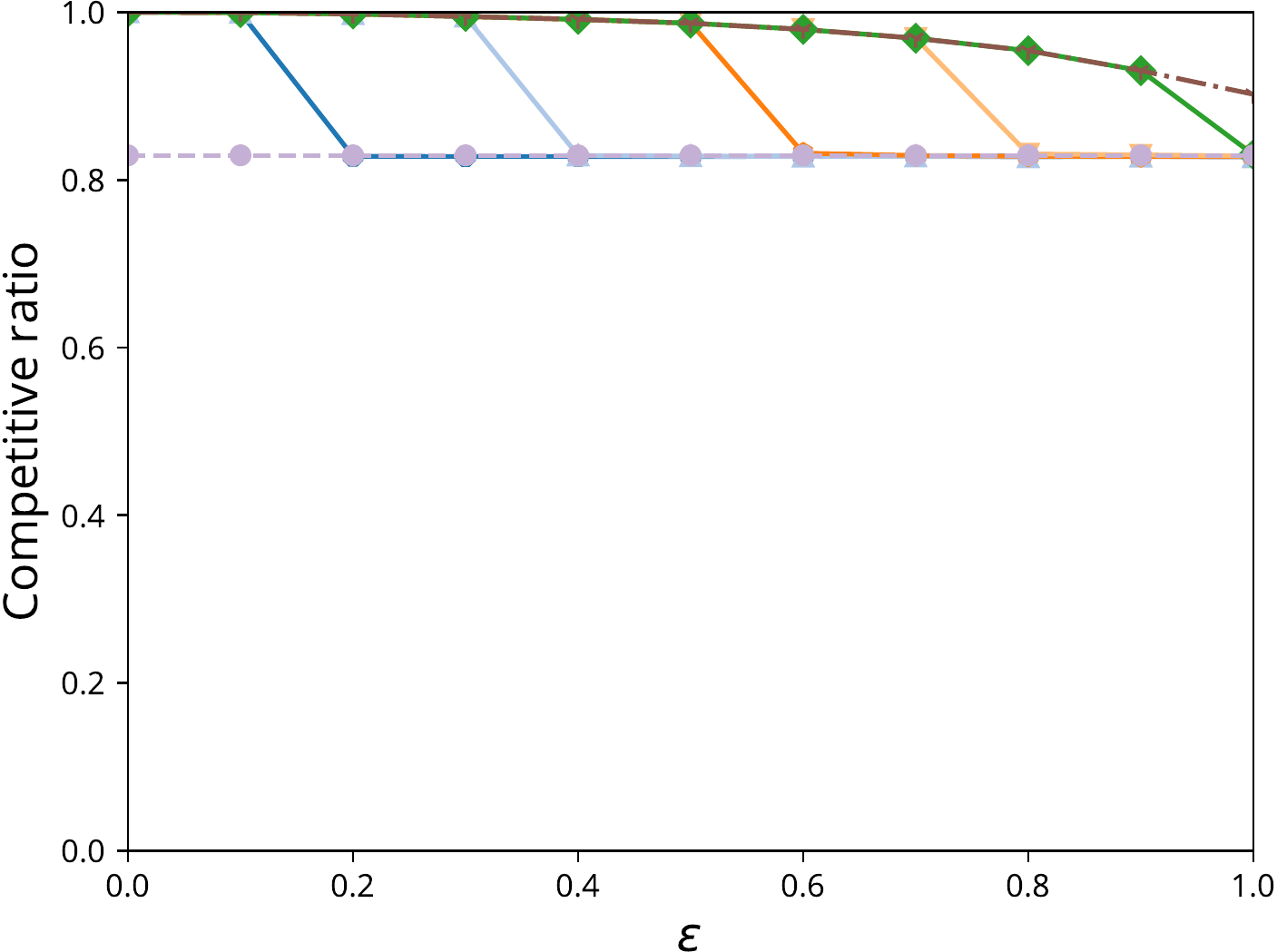}}
	\subfloat[][\textsf{Adversarial}, $k=50$]{\includegraphics[height=.25\hsize]{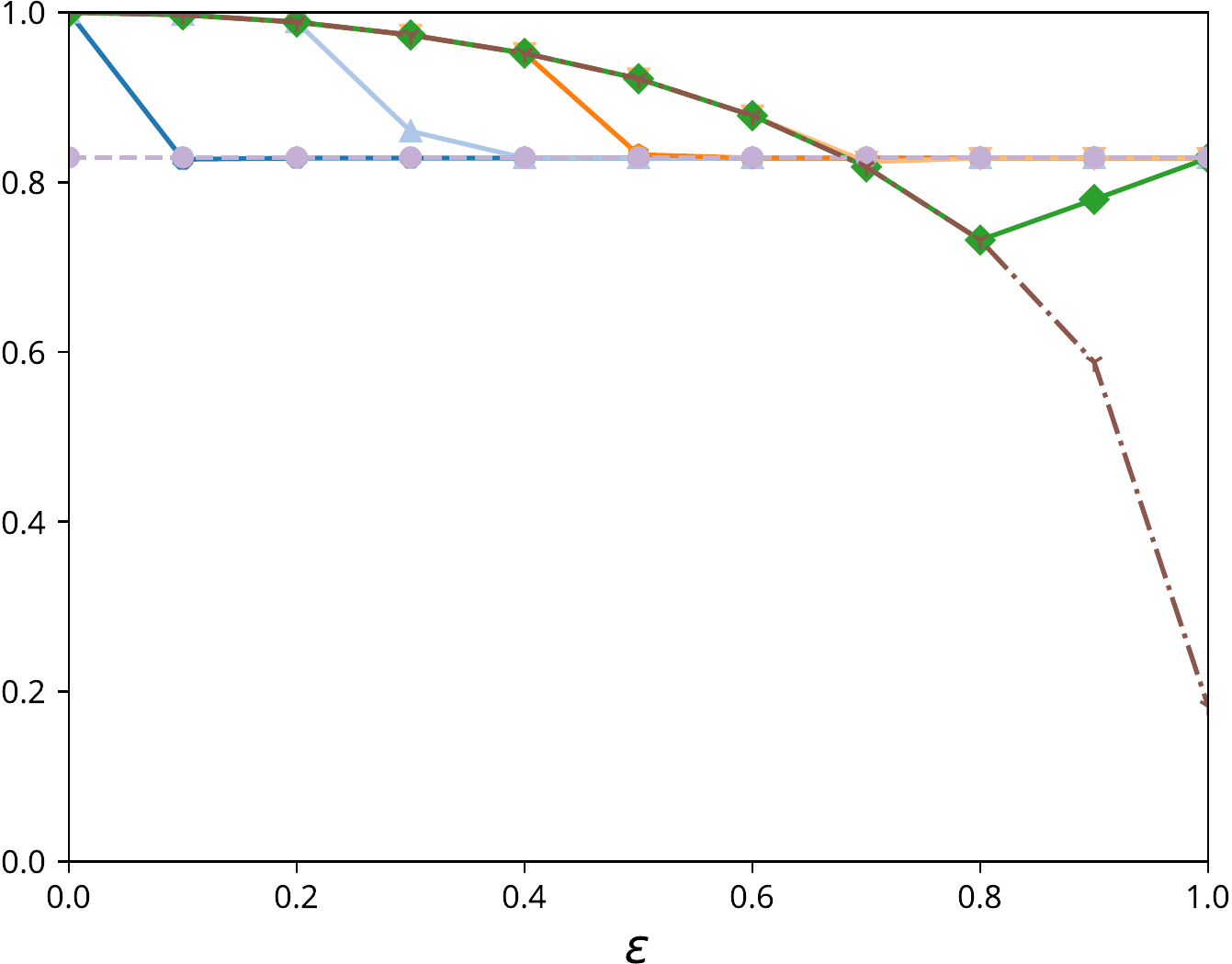}}
	\subfloat[][\textsf{Almost-Constant}, $k=50$]{\includegraphics[height=.25\hsize]{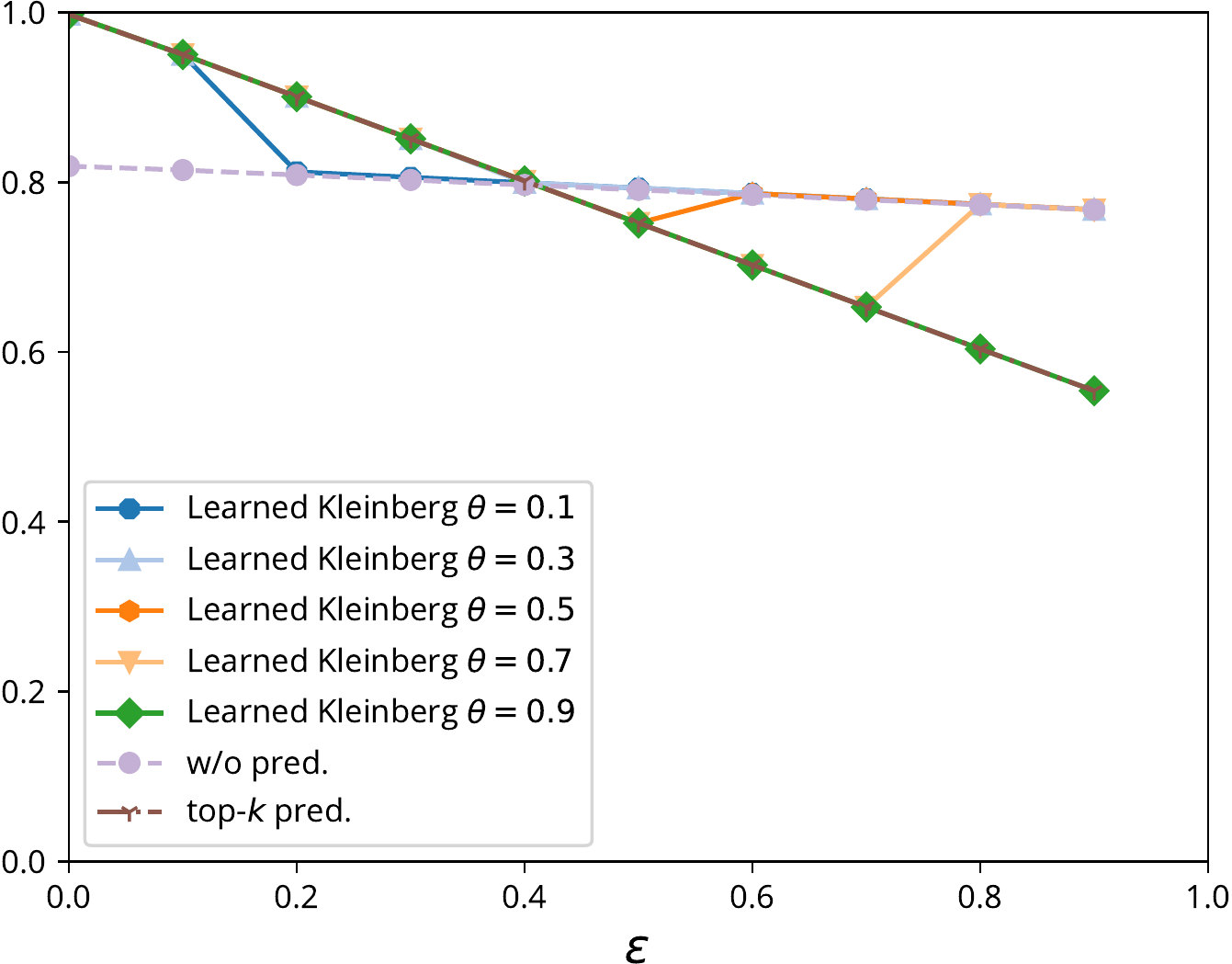}}
	\caption{Competitive ratios for the classical setting ($k=1$) and multiple-choice settings ($k=10,50$). We compare the proposed and existing algorithms for three types of datasets: \textsf{Uniform}, \textsf{Adversarial}, and \textsf{Almost-Constant}.}\label{fig:pathological}
\end{figure*}

\subsection{Experimental results}

\Cref{fig:pathological} shows the competitive ratio for \textsf{Uniform}, \textsf{Adversarial}, and \textsf{Almost-Constant}, respectively.
For the classical setting, as the theoretical results guarantee, we can observe that the competitive ratio of learned Dynkin is $1$ for $\epsilon = 0$ and never worse than the guaranteed competitive ratio $0.215$ for all $\epsilon$ if we set $\theta \approx 0.646$.

By comparing the proposed algorithms with top-$k$ pred.\ and w/o pred., we can observe that learned Dynkin and Kleinberg almost completely follow the predictions (top-$k$ pred.) for $\epsilon \le \theta$ and show competitive performance with Dynkin's or Kleinberg's algorithm (w/o pred.) for $\epsilon > \theta$.
For \textsf{Uniform}, since candidates with top-$k$ predictions tend to have high actual values, setting $\theta$ higher leads to better results.
On the other hand, \textsf{Almost-Constant}, setting $\theta$ smaller leads to better results.
For \textsf{Adversarial}, top-$k$ pred.\ performs better than w/o pred.\ for small $\epsilon$ and w/o pred.\ performs better than top-$k$ pred.\ for large $\epsilon$.
We can see that learned Dynkin with $\theta = 0.7$ has the best of both worlds, that is, it is competitive with top-$k$ pred.\ for small $\epsilon$ and competitive with w/o pred.\ for large $\epsilon$.
In general, it is difficult to decide an optimal $\theta$ before observing the actual values, but we can guarantee the worst-case competitive ratio by setting $\theta$ as the theorems suggest.

\paragraph{Comparison with AGKK algorithm.}
Here we compare learned Dynkin with AGKK algorithm.
For \textsf{Uniform} and \textsf{Adversarial} with small $\epsilon$, while learned Dynkin with any $\theta$ hires the optimal candidate, AGKK algorithm fails to hire the optimal one.
Note that learned Dynkin with $\theta \approx 0.646$ has the same worst-case guarantee as AGKK algorithm with $c = 1.71$.
Hence, we can claim that learned Dynkin has an advantage in its performance for small $\epsilon$ compared to AGKK algorithm with the same worst-case guarantee.
In general, as we can observe in our experimental results, which one is better depends on the datasets and the parameters $\theta$ and $\lambda$.

\paragraph{Comparison with the prophet secretary algorithm.}
Next, we compare learned Dynkin algorithm with the prophet secretary algorithm.
The prophet secretary algorithm performs poorly for \textsf{Uniform}.
For \textsf{Adversarial} and \textsf{Almost-Constant}, the prophet secretary algorithm performs well for some $\epsilon$, but its performance is unstable.
A possible reason is that the prophet secretary algorithm is sensitive to an error in the given distributions.
The predictions of \textsf{Uniform} are generated by adding the uniform distribution, but the posterior distributions of the actual values given the predictions are not the uniform distribution.
This error in the given distribution seems to significantly worsen the performance of the prophet secretary algorithm.
We conclude that the prophet secretary algorithm does not fit the setting with adversarial predictions.

\section{Concluding remarks}\label{sec:concluding-remarks}
In this study, we considered the classical and multiple-choice secretary problems with predictions.
We proposed algorithms that perform almost optimally if the predictions are accurate and competitively with the standard algorithms that do not use the predictions if the predictions are erroneous.
For the classical setting, we showed that no algorithm can perform optimally if the predictions are accurate while keeping the worst-case performance competitive with the classical secretary algorithm.

A possible direction for future research is to improve the competitive ratios.
For the classical setting, there is a gap between the lower bound $\max\{0.215,1-O(\epsilon)\}$ and upper bound $\max\{0.348,1-C\epsilon\}$ for any constant $C > 0$.
For the multiple-choice setting, since the optimal competitive ratio for the setting without predictions is $1-\Theta(1/\sqrt{k})$, it is interesting to consider removing the $\ln k$ term from the competitive ratio $1-O(\min\{\ln k/\sqrt{k}, \epsilon\})$.

Another direction is to consider combinatorial constraints such as matroid or knapsack constraints.
\citet{AGKK20} considered online bipartite matching and graphic matroids constraints in a slightly different setting.
Since the decision maker is given more information in our model, it might be possible to obtain better competitive ratios for these constraints.

\bibliographystyle{plainnat}
\bibliography{main}


\end{document}